\documentclass[runningheads,a4paper]{llncs}
\usepackage{times}
\usepackage{cite}
\usepackage{listings}
\usepackage{colonequals}
\usepackage{xspace}
\usepackage{amssymb}
\usepackage{graphicx}
\usepackage{stmaryrd}
\usepackage{cite}
\usepackage{framed}
\usepackage{comment}
\usepackage{version}
\usepackage{mathtools}

\usepackage{xcolor}

\usepackage{amsmath}    % need for subequations
\usepackage{booktabs}
\usepackage{url}
\urldef{\mailsa}\path|firstname.lastname@epfl.ch|

\excludeversion{shortv}
\includeversion{longv}

\usepackage{defs}

\usepackage{program}

\newcommand{\cvc}{\textsc{cvc}{\small 4}\xspace}

\newcommand{\ziii}{\textsc{z}{\small 3}\xspace}
\newcommand{\teq}{\approx}

\newcommand{\I}{\mathcal{I}}

\newcommand{\mods}{\mathbf{I}}
\newcommand{\lan}{\mathbf{L}}
\newcommand{\qlan}[1]{\mathcal{Q}( #1 )}

\newcommand{\modsof}[2]{{\llbracket {#1} \rrbracket}_{#2}}

\newcommand{\con}[1]{\mathsf{#1}}

\newcommand{\Int}{\con{Int}}
\newcommand{\Real}{\con{Real}}
\newcommand{\ite}{\con{ite}}

\newcommand{\tra}{\ensuremath{\mathrm{RA}}\xspace}
\newcommand{\tia}{\ensuremath{\mathrm{IA}}\xspace}
\newcommand{\lra}{\ensuremath{\mathrm{LRA}}\xspace}
\newcommand{\lia}{\ensuremath{\mathrm{LIA}}\xspace}
\newcommand{\larel}{\bowtie}

\newtheorem{thm}{Theorem}

\newtheorem{lem}{Lemma}
\newtheorem{defn}{Definition}

\newcommand{\procslra}{\mathcal{S}_{\lra}}
\newcommand{\procslrah}{S_R}
\newcommand{\procslrahh}{S_{R0}}
\newcommand{\procsslra}{\mathcal{S}_{S\lra}}
\newcommand{\procslia}{\mathcal{S}_{\lia}}
\newcommand{\procsliah}{S_I}
\newcommand{\procsliahh}{S_{I0}}

\newcommand{\opdivd}{\small\mathsf{div}}

\newcommand{\opdivu}{\small\mathsf{div}^{+}}
\newcommand{\opmod}{\xspace \mathsf{mod} \xspace}
\newcommand{\oplcm}{\xspace \mathsf{lcm} \xspace}
\newcommand{\opdivides}{\xspace \mid \xspace}

\newcommand{\smtsolve}{\small\mathsf{SMTQI}}

\newcommand{\Set}[1]{\left\{#1\right\}}
\newcommand{\parens}[1]{\left(#1\right)}

\newcommand{\bigand}{\bigand}

\newcommand{\fvars}{FV}

\newcommand{\normalize}[1]{{#1}\!\!\downarrow}

\DeclareMathOperator{\transform}{\ \leadsto \ }

\begin{document}
\sloppy

\mainmatter
\newcommand{\mytitle}{An Instantiation-Based Approach for Solving Quantified Linear Arithmetic}
\title{\mytitle}
\titlerunning{\mytitle}
\author{Andrew Reynolds\inst{1} \and Tim King\inst{2} \and Viktor Kuncak\inst{1}}

\authorrunning{Reynolds, King, Kuncak}

\institute{
  {\'E}cole Polytechnique F{\'e}d{\'e}rale de Lausanne (EPFL), Switzerland \and
  Verimag
}
\maketitle

%at least 70 and at most 150 words
\begin{abstract}
This paper presents a framework to derive instantiation-based decision procedures for satisfiability 
of quantified formulas in first-order theories, including its correctness, implementation, and evaluation.
Using this framework we derive 
decision procedures for linear real arithmetic ($\lra$) and linear integer arithmetic ($\lia$) formulas
with one quantifier alternation.
Our procedure can be integrated into the solving architecture used by typical SMT solvers.
Experimental results on standardized benchmarks from model checking, static analysis, and synthesis show that our implementation of the procedure in the SMT solver \cvc
outperforms existing tools for quantified linear arithmetic.
\end{abstract}

\section{Introduction}

%For certain background theories, 
%the satisfiability of quantified formulas can be established by
%reducing them to an equisatisfiable set of ground set of formulas
%and invoking a ground decision procedure on the result.

Among the biggest challenges in automated reasoning is
efficient support for quantifiers in the presence of
background theories.  Quantifiers enable direct encoding of a number of problems of interest, including
synthesis of software fragments from specifications
\cite{KuncakETAL10CompleteFunctionalSynthesis,ReynoldsDKBT15Cav,DBLP:conf/issac/SturmT11},
construction of transfer functions for program analysis
\cite{DBLP:conf/popl/Monniaux09}, invariant inference
\cite{DBLP:conf/pldi/GrebenshchikovLPR12,DBLP:conf/cade/BjornerMR12},
as well as analysis of properties that go beyond safety
\cite{DBLP:conf/cav/BeyenePR13,DBLP:conf/popl/BeyeneCPR14}.

The most commonly used complete method for deciding constraints over quantified theories
is \emph{quantifier elimination} \cite[Section
2.7]{Hodges93ModelTheory}. Quantifier elimination algorithms typically solve a more general problem,
of transforming arbitrary quantified formula with free variables into a theory-equivalent formula 
with no quantifiers.
% whereas preserving (or reducing) the set of free variables.
However, depending on the particular variant of the language of constraints, 
performing actual 
quantifier elimination can have worse complexity than
the decision problem \cite{Berman80ComplexityLogicalTheories},
%see also \cite[Lecture 24]{Kozen80ComplexityBooleanAlgebras}, 
in part because it is required to give an answer on any formula, and the smallest formula resulting from quantifier
elimination can be very large \cite{Weispfenning1997}. When
the goal is to decide the satisfiability of quantified
constraints, quantifier elimination may be doing unnecessary work.
More importantly, procedures based on quantifier elimination often do not handle the underlying ground constraints
in the most efficient way.
Thus, quantifier elimination tends to be prohibitively expensive in practice.
Recent work involving quantifier elimination~\cite{Monniaux10QuantifierEliminationLazyModelEnumeration,Bjoerner10LinearQuantifierEliminationAsAbstractDecision}
has been motivated by avoiding worst-case performance
by effectively computing an equisatisfiable set of ground formulas in a lazy fashion.

In the broader scope of automated theorem proving,
it is often important to reason about formulas involving multiple theories, 
each of which may or may not support quantifier elimination.
In practice, the goal is to obtain a framework for handling quantified formulas that is 
both complete for formulas belonging to decidable logics,
and empirically effective when completeness guarantees are not known.
To this end, modern SMT solvers most commonly use heuristic instantiation-based approaches~\cite{Detlefs03simplify:a},
which are incomplete but work well in practice for undecidable fragments of first-order logic.

Thus, our motivation is to capitalize 
both on recent advances in specialized techniques for quantified linear arithmetic~\cite{Bjoerner10LinearQuantifierEliminationAsAbstractDecision,DBLP:conf/cade/PhanBM12,komuravelli2014smt,bjornerplaying},
and recent advances in instantiation-based theorem proving for first-order logic~\cite{ganzinger2003new,MouraBjoerner07EfficientEmatchingSmtSolvers,reynolds14quant_fmcad}.
This paper seeks to bridge the gap between these two lines of research by
introducing an approach for establishing the satisfiability of formulas in quantified linear arithmetic 
based on a new \emph{quantifier instantiation} framework.
The use of quantifier instantiation for this task is motivated by the following.
\begin{itemize}
\item 
Procedures based on lazy quantifier instantiation
typically establish satisfiability much faster 
than their theoretical complexity.
\item 
Using quantifier instantiation for decidable fragments 
enables a uniform integration and composition 
with existing instantiation-based techniques~\cite{Detlefs03simplify:a,MouraBjoerner07EfficientEmatchingSmtSolvers,reynolds14quant_fmcad},
which are widely used by modern SMT solvers.
\item
An important class of synthesis problems can be expressed as
quantified formulas with one quantifier alternation.
As shown in~\cite{ReynoldsDKBT15Cav}, solutions for these problems
can be extracted from an unsatisfiable core of quantifier instantiations.
\end{itemize}

\sparagraph{Related Work}
Quantifier elimination has been used to, e.g., show decidability and classification of boolean
algebras \cite{Skolem19Untersuchungen,Tarski49ArithmeticalClassesTypesBooleanAlgebras},
Presburger arithmetic
\cite{Presburger29UeberVollstaendigkeitSystemsAritmethikZahlen},
decidability of products 
\cite{Mostowski52DirectProductsTheories,
  FefermanVaught59FirstOrderPropertiesProductsAlgebraicSystems},
\cite[Chapter 12]{Malcev71MetamathematicsAlgebraicSystems},
and algebraically closed fields
\cite{Tarski49ArithmeticalClassesTypesAlgebraicallyClosed}.
The original result on decidability of Presburger arithmetic
is by Presburger
\cite{Presburger29UeberVollstaendigkeitSystemsAritmethikZahlen}.
The space bound for Presburger arithmetic
was shown in
\cite{FerranteRackoff79ComputationalComplexityLogicalTheories}.
The matching lower and upper bounds for Presburger arithmetic were shown 
in \cite{Berman80ComplexityLogicalTheories},
see also \cite[Lecture 24]{Kozen06TheoryComputation}.
An analysis parameterized by the number of quantifier alternations
is presented in
\cite{ReddyLoveland78PresburgerBoundedAlternation}.
A mechanically verified quantifier elimination algorithm was developed
by Nipkow \cite{nipkow2008linear}.

An approach for lazy quantifier elimination for linear real arithmetic
was developed by Monniaux~\cite{Monniaux10QuantifierEliminationLazyModelEnumeration}.
Integration of linear quantifier elimination into the solving algorithm used by SMT solvers was developed in~\cite{Bjoerner10LinearQuantifierEliminationAsAbstractDecision},
though the presented integration is not model driven.
A lazy approach for quantifier elimination, which relies on an operation called model-based projection, has been developed
in the context of SMT-based model checking~\cite{komuravelli2014smt},
and can be used for extracting skolem functions for simulation synthesis~\cite{fedyukovichautomated}.
A recent approach for quantified formulas with arbitrary alternations has been developed by Bjorner~\cite{bjornerplaying} 
for several background theories, which is not based on instantiation.
The most widely used techniques for quantifier instantiation in SMT were developed in~\cite{Detlefs03simplify:a},
and later in~\cite{MouraBjoerner07EfficientEmatchingSmtSolvers,GeBarrettTinelli07SolvingQuantifiedVerificationConditionsUsingSatisfiability},
which primarily focused on uninterpreted functions.
% TIM: I think saying equisatisfiable is going to mislead people. I am not saying the claim is wrong, but is the audience going to unpack the definition of equisatisfiable to think through this sentence or are they going to think it is like other "equisatisfiable" things they have seen before, like tseitin encoding or unconstrained variable removal. My hunch is audiences will go towards the later. I am in favor of unpacking the definition explicitly:
% My Suggestion: Our approach for quantified linear arithmetic instantiates a quantified formula based on a candidate model. The approach terminates with either a finite set of instantiations that are satisfiable iff the original formula is satisfiable.
% TIM: Thoughts?
% AJR: Changed the sentence below to be more general.  MBQI in Z3 doesn't quite fit the criteria of your sentence, since its termination can be dependent upon guesses in model generation for UF, hence it is not based on the equivalences of instances <=> original formula.
%Our approach for quantified linear arithmetic 
%identifies a finite set of instantiations that is equisatisfiable to the quantified formula based on a stream of candidate models.
%
% cite bitvectors as an example of quantifier instantiation in SMT for theories?
Our approach for quantified linear arithmetic instantiates quantified formulas based on a lazy stream of candidate models,
terminating when either it finds a finite set of instances are unsatisfiable, or discovers that the original formula is satisfiable.
Other approaches in this spirit have been used to decide essentially uninterpreted fragment~\cite{GeDeM-CAV-09},
and, more generally, theories having a locality property~\cite{Jacobs09,bansaldeciding}; these works do not directly apply to quantified linear arithmetic.
A recent approach for quantified formulas with one quantifier alternation has been developed in the SMT solver Yices~\cite{dutertresolving},
which does not treat linear integer arithmetic.
The present paper builds upon our previous work for solving synthesis conjectures
using quantifier instantiation in SMT~\cite{ReynoldsDKBT15Cav},
where an approach for quantified linear arithmetic was described without a specific method for selecting instances and without 
completeness guarantees.
While the present paper focuses on linear arithmetic, where it outperforms existing approaches, we expect the presented framework to be relevant
for other quantified theories. 
\begin{longv}
Among the examples of further decidable quantified constraints are
quantified theories of
term algebras \cite[Chapter
23]{Malcev71MetamathematicsAlgebraicSystems},
\cite{Maher88CompleteAxiomatizationsAlgebrasTrees,
  SturmWeispfenning02QuantifierEliminationTermAlgebras}
and their extensions 
\cite{ComonDelor94EquationalFormulaeMembershipConstraints,
RybinaVoronkov01DecisionProcedureTermAlgebrasQueues,
KuncakRinard03StructuralSubtypingNonRecursiveTypesDecidable},
feature trees 
\cite{Backofen95CompleteAxiomatizationTheoryFeatureArityConstraints,
Treinen97FeatureTreesArbitraryStructures}, and
monadic second-order theories 
\cite{Walukiewicz02MonadicSecondOrderLogicTreeLikeStructures}.
\end{longv}
\begin{shortv}
Among the examples of non-trivial decidable quantified constraints are
algebraic data types (term algebras) \cite[Chapter
23]{Malcev71MetamathematicsAlgebraicSystems},
\cite{Maher88CompleteAxiomatizationsAlgebrasTrees,
  SturmWeispfenning02QuantifierEliminationTermAlgebras}.
\end{shortv}
% TIM: It is unclear what is being claimed in the previous sentence here. It sounds like you suspect that the technique is applicable to these constraints. I kinda doubt this is being claimed. Even if this is the claim, it seems like dialing the number of examples and citations here is worth the space. Also as these are BIG claims they feel like they should have more discussion if they are being made.

\sparagraph{Contributions}
This paper makes the following contributions.
\begin{shortv}
First, we define a general class of instantiation-based procedures
for establishing the satisfiability of quantified formulas in Section~\ref{sec:qi}.
\end{shortv}
\begin{longv}
First, we define a general class of instantiation-based procedures
for establishing the satisfiability of quantified formulas in Section~\ref{sec:qi},
and show that these procedures can be used in part as an approach for solving
synthesis problems in Section~\ref{sec:synth}.
\end{longv}
We demonstrate instances of the procedure are sound and complete 
for formulas over linear real arithmetic ($\lra$) and linear integer arithmetic ($\lia$) with one quantifier alternation
in Sections~\ref{sec:lra} and~\ref{sec:lia}, 
two quantified fragments for which many current SMT solvers do not have efficient support for.
%two quantified fragments for which procedures have not been integrated efficiently into SMT solvers before.
We show how our procedure can be integrated into the solving architecture used by SMT solvers in Section~\ref{sec:smt}. 
%We present an evaluation of an implementation of the procedures for $\lia$ and $\lra$ in the SMT solver \cvc, showing that it
%outperforms state-of-the-art SMT solvers and theorem provers for quantified linear arithmetic benchmarks in Section~\ref{sec:results}.
To our knowledge, our approach is the first complete algorithm for quantified linear arithmetic with one alternation
that is based purely on quantifier instantiation,
which has the advantage of being composable with existing techniques and whose soundness is straightforward to verify.
In Section~\ref{sec:results},
we demonstrate an implementation of the procedures for $\lia$ and $\lra$ in the SMT solver \cvc,
which in addition to having the aforementioned advantages,
outperforms state-of-the-art SMT solvers and theorem provers for quantified linear arithmetic benchmarks.

\subsection{Preliminaries}

We consider formulas in multi-sorted first-order logic.
A \emph{signature} $\Sigma$ consists of 
a countable set of sort symbols and
a set of function symbols.
Given a signature $\Sigma$,
well-sorted terms, atoms, literals, and formulas
are defined as usual, and referred to respectively as \emph{$\Sigma$-terms}.
We denote by $\fvars( t )$ 
the set of free variables occurring in the term $t$, and extend this notion to formulas.
A $\Sigma$-term or formula is \emph{ground} if it has no free variables.
% TIM: I changed this to a version that removes $\vec k$ as an orphan.
%A term written $t[ \vec k ]$ denotes a term whose free variables are a subset of $\vec k$.
A term written $t[ \vec k ]$ denotes a term whose free variables are in $\vec k$.

A \emph{$\Sigma$-interpretation $\I$} maps
\begin{itemize}
\item each set sort symbol $\sigma \in \Sigma$ to a non-empty set $\sigma^\I$,
the \emph{domain} of $\sigma$ in $\I$, and 
\item each
function $f \in \Sigma$ of sort $\sigma_1 \times \ldots \times \sigma_n \rightarrow \sigma$
to a total function $f^\I$ of sort $\sigma^\I_1 \times \ldots \times \sigma^\I_n \rightarrow \sigma^\I$
where $n > 0$, and 
to an element of $\sigma^\I$ when $n = 0$.
\end{itemize}
We write $t^\I$ to denote the interpretation of $t$ in $\I$, defined inductively as usual.
A satisfiability relation between $\Sigma$-interpretations and 
$\Sigma$-formulas, written $\I \models \varphi$, is also defined inductively as usual. In particular, we assume that $\I \models \lnot \varphi$ if and only if it is not the case
that $\I \models \varphi$.
We say that $\I$ is \emph{a model of $\varphi$} if $\I$ satisfies $\varphi$.
Formulas $\varphi_1$ and $\varphi_2$ are \emph{equivalent (up to $\vec k$)} 
if they are satisfied by the same set of models (when restricted to the interpretation of variables $\vec k$).
%if they have the same set of models restricted to $\vec k$.

A \emph{theory} is a pair $T = (\Sigma, \mods)$ where 
$\Sigma$ is a signature and $\mods$ is a non-empty set of $\Sigma$-interpretations,
the \emph{models} of $T$.
We assume $\Sigma$ contains the equality predicate, which we denote by $\teq$.
Let $\modsof{\varphi}{T}$ denote the set of $T$-models of $\varphi$. 
Observe that $\modsof{\lnot \varphi}{T} = \mods \setminus \modsof{\varphi}{T}$.
A $\Sigma$-formula $\varphi[\vec x]$ is 
\emph{$T$-satisfiable}  
if it is satisfied by some interpretation in $\mods$ 
(i.e.\ $\modsof{\varphi}{T} \neq \emptyset$).
Dually,
a $\Sigma$-formula $\varphi[\vec x]$ is 
\emph{$T$-unsatisfiable}  
if it is satisfied by no interpretation in $\mods$ 
(i.e.\ $\modsof{\varphi}{T} = \emptyset$).
A formula $\varphi$ is \emph{$T$-valid} if every model of $T$ is a model of $\varphi$
(i.e., $\modsof{\varphi}{T} = \mods$).
Given a fragment $\lan$ of the language of $\Sigma$-formulas,
a $\Sigma$-theory $T$ is \emph{satisfaction complete with respect to $\lan$}
if every closed $T$-satisfiable formula of $\lan$ is $T$-valid. 
In terms of set of models,
satisfaction completeness means that $\modsof{\varphi}{T} \neq \emptyset$ implies 
\begin{shortv}$\modsof{\varphi}{T}=\mods$.\end{shortv}
\begin{longv}$\modsof{\varphi}{T}=\mods$,
or, in other words, for every $F \in \lan$ exactly one of the following two cases hold:
$\modsof{\varphi}{T} =\emptyset$, or $\modsof{\varphi}{T} = \mods$. 
If additionally $\lan$ is closed under negation, then, for every $\varphi \in \lan$,
either $\varphi$ or $\lnot F$ is unsatisfiable.
\end{longv}
%TIM: I dropped "(dually, either $F$ or $\lnot F$ is valid)". Does this add much?
%
% IF the set of models is defined by recursive set of axioms, in that case $T$-satisfiability of $\lan$ formulas
%is necessarily decidable by semi-decision procedure for first-order logic.

A set $\Gamma$ of formulas \emph{$T$-entails} a $\Sigma$-formula $\varphi$,
written $\Gamma \models_T \varphi$,
if every model of $T$ that satisfies all formulas in $\Gamma$ satisfies $\varphi$ as well.
A set of literals $M$ \emph{propositionally entails} a formula $\varphi$, written $M \models_p \varphi$, 
if $M$ entails $\varphi$
when considering all atomic formulas in $M \cup \varphi$ as propositional variables;
such entailment is one of from propositional logic and is independent of the theory.

We write $\tra$ (resp. $\tia$) to denote the theory of real (resp. integer) arithmetic.
Its signature consists of the sort $\Real$ (resp. $\Int$), the binary predicate symbols $>$ and $<$,
functions $+$ and $\cdot$ denoting addition and multiplication,
and the constants of its sort interpreted as usual.
We write $t \leq s$ as shorthand for $\lnot ( t > s )$,
and $t \geq s$ as shorthand for $\lnot ( t < s )$.
We write $\lra$ (resp. $\lia$) to denote the language of linear real (resp. integer) arithmetic formulas,
that is, whose literals are of the form $(\lnot)( c_1 \cdot x_1 + \ldots + c_n \cdot x_n \larel c )$
where $c_1, \ldots, c_n, c$ and $x_1, \ldots, x_n$ are non-zero constants and distinct variables of sort $\Real$ (resp. $\Int$) respectively,
and $\larel$ is one of $>$, $<$, or $\teq$.
%TIM: I changed a number of instances from atom to literal. All of the discussion and formulas seem to be about literals.
For each literal of this form, there exists an equivalent literal that is in \emph{solved form with respect to $x_i$} for each $i= 1, \ldots, n$.
That is, an $\lra$-literal is in solved form with respect to $x$ if it is of the form $(\lnot)( x \larel t )$, where $x \not\in \fvars( t )$.
Similarly,
an $\lia$-literal is in solved form with respect to $x$ if it is of the form $(\lnot)( c \cdot x \larel t )$, where $x \not\in \fvars( t )$ and $c$ is an integer constant greater than zero.
For integer constants $c_1$ and $c_2$ and non-zero constant $c$, we write $c_1 \equiv_c c_2$ to denote that $c_1$ and $c_2$ are congruent modulo $c$,
that is $(c_1 \ \opmod \ c) = (c_2 \ \opmod \ c)$, and
we write $c \mid c_1$ if $c$ divides $c_1$.
%TIM: I think we should be a bit briefer here. How is the above?
%A non-zero constant $c_1$ divides a constant $c_2$, written $c_1 \mid c_2$, when $c_2 \equiv_{c_1} 0$.
% AJR: The current sentence is introducing both the congruence relation and divides, not sure how to make this more concise while since introducing both concepts.

\section{Quantifier Instantiation for Theories}
\label{sec:qi}

In this section, we assume a fixed theory $T$ and a language $\lan$ that is closed under negation
and such that the satisfiability of finite sets of $\lan$ formulas modulo $T$ is decidable. We
present a procedure for checking satisfiability of formulas in the language $\qlan{\lan} =
\{ \forall \vec x\, \varphi[\vec a, \vec x] \mid \varphi[\vec a, \vec x] \in \lan \}$.

\begin{figure}[t]
\begin{framed}
$\mathcal{P}_{\mathcal{S}}( \exists \vec a\, \forall \vec x\, \varphi[\vec a, \vec x] )$:
\begin{enumerate}
\item[\ ] Let $\Gamma := \emptyset$ and $\vec{\con k}$, $\vec{\con e}$ be tuples of fresh constants of the same type as $\vec a$, $\vec x$.
\item[\ ] Repeat
 \begin{itemize}
  \item[\ ] If $\Gamma$ is $T$-unsatisfiable, then return ``unsat".
  \item[\ ] If $\Gamma' = \Gamma \cup \{ \lnot \varphi[\vec{\con k}, \vec{\con e}] \}$ is $T$-unsatisfiable, then return ``sat".
  \item[\ ] Otherwise,
\begin{itemize}
  \item[\ ] Let $\I$ be a model of $T$ and $\Gamma'$ and let $\vec{t}[\vec{\con k}] = \mathcal{S}( \I, \Gamma, \lnot \varphi[\vec{\con k}, \vec{\con e}] )$.
  \item[\ ] $\Gamma := \Gamma \cup \{ \normalize{\varphi[\vec{\con k}, \vec{t}[\vec{\con k}]]} \}$.
\end{itemize}  
\end{itemize}
\end{enumerate}
\vspace*{-2ex}
\end{framed}
\vspace{-2ex}
\caption{An instantiation-based approach $\mathcal{P}_{\mathcal{S}}$
for determining the $T$-satisfiability of $\exists \vec a\, \forall x\, \varphi[\vec a, x]$
parameterized by selection function $\mathcal{S}$.
\label{fig:proc-qi}}
\end{figure}

\subsection{An Instantiation Procedure and Its Soundness}

Figure~\ref{fig:proc-qi} presents an instantiation-based approach for determining the satisfiability 
of a $T$-formulas $\exists \vec a\, \forall \vec x\, \varphi[\vec a, \vec x]$,
where $\varphi[\vec a, \vec x]$ belongs to $\lan$.
The procedure introduces a tuple of distinct fresh constants
$\vec{\con k}$ of the same sort as $\vec{a}$,
and $\vec{ \con e }$ of the same sort as $\vec x$.
It maintains a set of formulas $\Gamma$, initially empty, 
and terminates when either $\Gamma$ or $\Gamma \cup \{ \lnot \varphi[\vec{\con k}, \vec{\con e}] \}$ is $T$-unsatisfiable.
On each iteration, 
the procedure invokes the subprocedure $\mathcal{S}$ (over which the procedure is parameterized),
which returns a tuple of terms $\vec{t}[\vec{\con k}]$ whose free variables are a subset of $\vec{ \con k }$.
We then add to $\Gamma$ the formula $\normalize{\varphi[\vec{\con k}, \vec{t}[\vec{\con k}]]}$,
a formula equivalent to $\varphi[\vec{\con k}, \vec{t}[\vec{\con k}]]$ up to $\vec{\con k}$\footnote{We further comment on examples of operators $\normalize{}$ in Sections~\ref{sec:lra} and~\ref{sec:lia}.}.
We call $\mathcal{S}$ the \emph{selection function} of $\mathcal{P}_{\mathcal{S}}$.

The intution of the algorithm is to
find a subset of the instances of $\forall \vec x\, \varphi[\vec a, \vec x]$ that are either
(a) unsatisfiable, and are thus sufficient for showing that $\forall \vec x\, \varphi[\vec a, \vec x]$ is unsatisfiable, or
(b) satifiable and entail $\forall \vec x\, \varphi[\vec a, \vec x]$.
The algorithm recognizes the latter case by checking the satisfiability of $\Gamma \cup \neg \varphi[ \vec k, \vec e]$ on each iteration of its main loop.
In either case, the algorithm may terminate before enumerating all instances of $\forall \vec x\, \varphi[\vec a, \vec x]$.
In practice, we have found the algorithm often terminates after enumerating only a small number of instances for benchmarks that occur in practice.

\begin{defn}
A selection function (for $\lan$) takes as arguments 
an interpretation $\I$,
a set of formulas $\Gamma$, and
a formula $\lnot \varphi[\vec{ \con k }, \vec{\con e}]$ in $\lan$, where
$\I \models \Gamma \cup \lnot \varphi[\vec{ \con k }, \vec{\con e}]$,
and returns a tuple of terms $\vec{t}[\vec{\con k}]$ such that 
$\normalize{\varphi[\vec{ \con k }, \vec{t}[\vec{\con k}]]}$ is also in $\lan$.
\end{defn}
Note that a selection function is only defined if $\I$ is a model for $T$, $\Gamma$ and $\lnot \varphi[\vec{ \con k }, \vec{\con e}]$.
We first show that the procedure always returns correct results, regardless of the behavior of the selection function, leaving the termination question for the next 
subsection.
Detailed proofs of the all claims in this paper can be found in the extended version of this report~\cite{RKK2015}.

\begin{lem}\label{lem:proc-qi-sound-unsat}
If $\mathcal{P}_{\mathcal{S}}$ terminates with ``unsat", 
then $\exists \vec a\, \forall \vec x\, \varphi[\vec a, \vec x]$ is $T$-unsatisfiable.
\end{lem}
\begin{longv}
\begin{proof}
In this case, 
there exists a set $\Gamma$ that is equivalent to
$\{\varphi[\vec{\con k}, \vec{t}_1], \ldots, \varphi[\vec{\con k}, \vec{t}_{p}] \}$ and is $T$-unsatisfiable
where $\vec{\con k}$ are distinct fresh constants.
Thus, $\forall \vec x\, \varphi[\vec{ \con k }, \vec x]$ is $T$-unsatisfiable.
Since $\vec{\con k}$ are distinct and fresh, we conclude that $\exists \vec a\, \forall \vec x\, \varphi[\vec a, \vec x]$ is $T$-unsatisfiable.
\qed
\end{proof}
\end{longv}

\begin{lem}\label{lem:proc-qi-sound-sat}
If $\mathcal{P}_{\mathcal{S}}$ terminates with ``sat", 
then $\exists \vec a\, \forall \vec x\, \varphi[\vec a, \vec x]$ is $T$-satisfiable.
\end{lem}
\begin{longv}
\begin{proof}
In this case,
there exists a set $\Gamma$ equivalent to 
$\{ \varphi[\vec{\con k}, \vec{t}_1], \ldots, \varphi[\vec{\con k}, \vec{t}_{p}] \}$
that is $T$-satisfiable, where $\vec{\con k}$ are distinct fresh constants,
and $\Gamma'$ equivalent to $\Gamma \cup \{ \lnot \varphi[\vec{\con k}, \vec {\con e}] \}$ that is $T$-unsatisfiable.
The variables $\vec{ \con e}$ do not occur in $\Gamma$,
since the only formulas added to it are of the form $\normalize{\varphi[\vec{\con k}, \vec{t}[\vec{\con k}]]}$.
Thus, we have that 
$\Gamma \cup \{ \exists \vec x\, \lnot \varphi[\vec{\con k}, \vec x] \}$ is $T$-unsatisfiable.
% Thus, $\Gamma \cup \{ \forall \vec a\, \exists \vec x\, \lnot \varphi[\vec a, \vec x] \}$ is $T$-unsatisfiable as well.
%Since $T$ is satisfaction-complete with respect to $\lan$, and at least one model must falsify $\forall \vec a\, \exists \vec x\, \lnot \varphi[\vec a, \vec x]$,
%this implies all models of $T$ falsify $\forall \vec a\, \exists \vec x\, \neg \varphi[\vec a, \vec x]$,
%and thus all models of $T$ satisfy $\exists \vec a\, \forall \vec x\, \varphi[\vec a, \vec x]$.
Let $\I$ be a model of $\Gamma$. Since $\I$ is not a model of $\Gamma'$,
it must be the case that $\I \not\models \exists \vec x\, \lnot \varphi[\vec{\con k}, \vec x]$,
and hence $\I$ is a model for $\exists \vec a\, \forall \vec x\, \varphi[\vec a, \vec x]$.
\qed
\end{proof}
\end{longv}

\subsection{Termination of the Instantiation Procedure}

The following properties of selection functions will be of interest.

\begin{defn}[Finite]
A selection function $\mathcal{S}$ is \emph{finite} for $\varphi[\vec{ \con k }, \vec{\con e}]$ if 
there exists a finite set $\mathcal{S}^\ast( \varphi[\vec{ \con k }, \vec{\con e}] )$ such that 
$\mathcal{S}( \I, \Gamma, \lnot \varphi[\vec{ \con k }, \vec{\con e}] ) \in \mathcal{S}^\ast( \varphi[\vec{ \con k }, \vec{\con e}] )$
for all $\I$, $\Gamma$.
\end{defn}

\begin{defn}[Monotonic]
% TIM: The typesetting of "monotonic for" looks strange in the pdf. Maybe just emphasize monotonic. Same for "finite for" and "model-preserving for".
A selection function $\mathcal{S}$ is \emph{monotonic} for $\varphi[\vec{ \con k }, \vec{\con e}]$ if
whenever $\Gamma \models \varphi[\vec{\con k}, \vec{t}]$, 
we have that $\mathcal{S}( \I, \Gamma, \lnot \varphi[\vec{ \con k }, \vec{\con e}] ) \neq \vec{t}$.
%$\mathcal{S}( \I, \Gamma \cup \{ \varphi[\vec{\con k}, \vec{t}] \}, \lnot \varphi[\vec{ \con k }, \vec{\con e}] ) \neq \vec{t}$
%for all $\vec{t}, \Gamma$ such that
%$\vec{ \con e } \not\in \fvars( \vec{t}, \Gamma )$.
\end{defn}
Observe that, if $\mathcal{S}$ is a monotonic selection function, then for any finite list of terms $\vec{t}_1, \ldots \vec{t}_n$ we have
$
\mathcal{S}( \I, \{ 
  \normalize{\varphi[\vec{\con k}, \vec{t_1}]}, 
  \ldots,
  \normalize{\varphi[\vec{\con k}, \vec{t_n}]} \}, \lnot \varphi[\vec{ \con k }, \vec{\con e}] ) \notin \{ \vec{t}_1, \ldots, \vec{t}_n\}
$.
\begin{defn}[Model-Preserving]
A selection function $\mathcal{S}$ is \emph{model-preserving} for $\varphi[\vec{ \con k }, \vec{\con e}]$ if whenever
$\mathcal{S}( \I, \Gamma, \lnot \varphi[\vec{ \con k }, \vec{\con e}] ) = \vec{t}$,
we have that $\I \models \lnot \varphi[\vec{ \con k }, \vec{t}]$.
\end{defn}

\begin{lem}\label{lem:sel-mm-m}
A selection function that is model-preserving for $\varphi[\vec{ \con k }, \vec{\con e}]$ is also monotonic for $\varphi[\vec{ \con k }, \vec{\con e}]$.
\end{lem}
\begin{longv}
\begin{proof}
Assume that $\mathcal{S}$ is model-preserving for $\varphi[\vec{ \con k }, \vec{\con e}]$ and that 
$\mathcal{S}( \I, \Gamma \cup \{ \varphi[\vec{\con k}, \vec{ t }] \}, \lnot \varphi[\vec{ \con k }, \vec{\con e}] ) = \vec{s}$.
By definition of selection function, we have that $\I \models \varphi[\vec{ \con k }, \vec{t}]$.
By definition of model-preserving, we have that $\I \models \lnot \varphi[\vec{ \con k }, \vec{s}]$.
Thus, $\vec{s} \neq \vec{t}$ and $\mathcal{S}$ is monotonic for $\varphi[\vec{ \con k }, \vec{\con e}]$.
\qed
\end{proof}
\end{longv}

\begin{thm}\label{lem:sel-fin-mono}
If $\mathcal{S}$ is finite and monotonic for $\varphi[\vec{ \con k }, \vec{\con e}]$ in $\lan$, 
then $\mathcal{P}_{\mathcal{S}}$ is a (terminating) decision procedure for the $T$-satisfiability of 
$\exists \vec a\, \forall \vec x\, \varphi[\vec a, \vec x]$.
\end{thm}
\begin{proof}
Given a monotonic and finite $\mathcal{S}$, the procedure $\mathcal{P}_{\mathcal{S}}$
can only execute a finite number of iterations.
Assuming a decision procedure for determining the $T$-satisfiability of $T$-formulas in $\lan$,
$\mathcal{P}_{\mathcal{S}}( \exists \vec a\, \forall \vec x\, \varphi[\vec a, \vec x] )$ must terminate.
By Lemmas~\ref{lem:proc-qi-sound-unsat} and~\ref{lem:proc-qi-sound-sat}, 
$\mathcal{P}_{\mathcal{S}}$ is a decision procedure for the $T$-satisfiability of $\exists \vec a\, \forall \vec x\, \varphi[\vec a, \vec x]$.
\qed
\end{proof}

\ 

In this paper we will identify selection functions $\mathcal{S}$ that are finite and monotonic
for all $\exists \vec a\, \forall \vec x\, \varphi[\vec a, \vec x]$ residing in fragments $\lan$.
The fragments we consider are satisfaction complete. We consider satisfaction completeness
to be a good guiding principle when choosing candidate logical fragments to which our method
can be applied successfully. 

\begin{longv}
\subsection{Connection to Synthesis}
\label{sec:synth}

%TIM: This section is nice as a motivation, but is it necessary? It seems to occupy a lot of space. Alternatively, is there a shorter explanation? Maybe <= 1/2 a page? Can we make this a subsection of Section 2? How about mentioning in related work?

The connection between quantifier elimination and synthesis has been shown fruitful in previous work~\cite{KuncakETAL10CompleteFunctionalSynthesis};
it is one of our motivations for further improving quantified reasoning modulo theories. The procedure mentioned in this section can be used to synthesize functions from certain classes of specifications.
%In this approach, we assume that $T$ is satisfaction complete for the language $\lan$.
Consider (second-order) $T$-formulas 
of the form:
\begin{eqnarray} \label{eqn:syn_conj}
\exists \vec f \:
\forall \vec x \: \varphi[\vec f, \vec x]
\end{eqnarray}
of the form $\exists \vec f \: \forall \vec x \: \varphi[\vec f, \vec x]$,
where $\varphi$ is a quantifier-free formula,
$\vec x = ( x_1, \ldots, x_n )$ is a tuple of variables of sort $\sigma_i$ for $i = 1, \ldots, n$, and
$\vec f = ( f_1, \ldots, f_m )$ is a tuple of functions of sort $\sigma_1 \times \ldots \times \sigma_n \rightarrow \tau_j$ for $j =1, \ldots, m$.
We call such formulas \emph{synthesis conjectures}.
A synthesis conjecture is \emph{single invocation} (over $\lan$) if it is equivalent to:
\begin{eqnarray} \label{eqn:syn_conj_si}
\exists \vec f \:
\forall \vec x \: \psi[\vec x, \vec f( \vec x )]
\end{eqnarray}
where $\psi[\vec x, \vec y] \in \lan$.
That is, functions from $\vec f$ are applied to the tuple $\vec x$ only.
The formula (\ref{eqn:syn_conj_si}) is equivalent to the (first-order) formula
$\forall \vec x \: \exists \vec y \: \psi[ \vec x,\vec y]$, whose negation
\begin{eqnarray} \label{eqn:syn_conj_si_fo}
\exists \vec x \: \forall \vec y \: \neg \psi[\vec x, \vec y]
\end{eqnarray}
is suitable as an input to Figure~\ref{fig:proc-qi}.
As observed in~\cite{ReynoldsDKBT15Cav}, solutions for single invocation synthesis conjectures  
can be extracted from an unsatisfiable core of instantiations when proving the unsatisfiability of (\ref{eqn:syn_conj_si_fo}).
In particular, let $\vec{ \con k }$ be a set of distinct fresh variables of the same sort as $\vec x$, and say the set
%If running the algorithm in Figure~\ref{fig:proc-qi} on $\forall \vec y \: \neg \psi[\vec{ \con k },\vec y]$
%determines that the set of instantiations
$\{ \normalize{\neg \psi[\vec{ \con k }, \vec{t}_1[ \vec{ \con k }] ]}, \ldots, \normalize{\neg \psi[\vec{ \con k },  \vec{t}_p[ \vec{ \con k }] ]} \}$
is $T$-unsatisfiable
where $\vec{t}_i = (t^1_i[ \vec{ \con k }], \ldots, t^m_i[ \vec{ \con k }] )$ for $i = 1, \ldots, p$.
Then:
\begin{equation} \label{eqn:syn_conj_si_sol}
1 \leq j \leq m : f_j = \lambda \vec x.\, \ite( \psi[\vec x, t^j_p[\vec x]], t^j_p[\vec x], (\,\cdots\, 
                                          \ite( \psi[\vec x, t^j_2[\vec x]], t^j_2[\vec x], t^j_1[\vec x] )  ))
\end{equation}
is a solution for $\vec f$ in (\ref{eqn:syn_conj_si}).
%When using the instantiation-based procedure from Section~\ref{sec:qi}
%as a subprocedure for discharging (\ref{eqn:syn_conj_si_fo}),
%the above approach is a sound and complete method for synthesizing tuples of functions whose specification is a
%single invocation synthesis conjecture over linear integer arithmetic.
We use the instantiation-based procedure in Figure~\ref{fig:proc-qi} for discharging (\ref{eqn:syn_conj_si_fo}).
In contrast to prior work~\cite{ReynoldsDKBT15Cav}, 
we here devise
selection functions $\mathcal{S}$ for $\lan$ that are finite and monotonic, 
obtaining a sound and complete method for synthesizing tuples of functions whose specification is a
single invocation synthesis conjecture over $\lan$.
In the following sections, we show such selection functions both for linear real arithmetic ($\lra$) and linear real arithmetic ($\lia$).

It is important to note that the solution (\ref{eqn:syn_conj_si_sol}) does not necessarily belong to the language $\lan$,
since there is no restriction on the selection functions for $\lan$ that restricts its return value $\vec{t}$ to terms in $\lan$.
For example, 
in our approach for linear real arithmetic, $\vec{t}$ may contain a free distinguished constant $\delta$ representing an infinitesimal positive value,
and in our approach for linear integer arithmetic, $\vec{t}$ may contain integer division.
Additional steps may be necessary in practice for making (\ref{eqn:syn_conj_si_sol}) a computable function.
\end{longv}

\begin{longv}
\subsection{Illustration: Instantiation for a Simple Fragment of $\lra$}

We first present a selection function for a restricted class $\lan$ of $\lra$-formulas $\exists \vec a\, \forall x, \varphi[ \vec a, x ]$,
namely whose universal quantifier is over single variable $x$ of sort $\Real$, $\vec a$ are variables of sort $\Real$,
and whose (skolemized) body $\varphi[ \vec{ \con k }, \con e ]$ is of the form:
\begin{equation} \label{eqn:quant-simp}
( \con e < \ell_1 \vee \ldots \vee \con e < \ell_n \vee \con e > u_1 \vee \ldots \vee \con e > u_m  )
\end{equation}
% TIM: I find the extra negation thoughout this section a bit confusing.
where at least one of $\{ n, m \}$ is greater than zero, and $\con e \not\in \fvars( \ell_1, \ldots, \ell_n, u_1, \ldots, u_m )$.
Figure~\ref{fig:sel-real-simp} gives a selection function for $\procsslra$.
It considers the interpretation of terms $\ell_1, \ldots, \ell_n$ and $u_1, \ldots, u_m$ in a model $\I$ of $\Gamma$.
If $n>0$, then $\procsslra$ returns the $\ell_j$ whose value is maximal in $\I$.
If $n=0$, then $m>0$ and $\procsslra$ returns the $u_i$ whose value is minimal in $\I$.

\begin{figure}[t]
\begin{framed}
$\procsslra( \I, \Gamma, \neg \varphi[ \vec{\con k}, \con e ] )$
\ \ \ where $\varphi[ \vec{\con k}, \con e ]$ is \ \ \ 
  $\displaystyle\bigvee_{i=1}^n e < l_i \lor 
                \bigvee_{i=1}^m e > u_i$, \ for $n > 0$ or $m > 0$
%\item[\ ] Let $\I$ be a model of $\lra$ and $\Gamma$.
%\item[\ ] =

\hspace{1cm}
Return
$
    \begin{cases}
    \ell_j, & \mbox{ if } n>0 \mbox{ and } \ \mathsf{max} \{ \ell_1^\I, \ldots, \ell_n^\I \} = \ell_j^\I \\
    u_i, &  \mbox{ if } n=0 \mbox{ and } \ \mathsf{min} \{ u_1^\I, \ldots, u_m^\I \} = u_i^\I 
    \end{cases}
$
%\end{itemize}
\end{framed}
\vspace{-4ex}
\caption{A selection function $\procsslra$ for a simple fragment of $\lra$.
\label{fig:sel-real-simp}}
\end{figure}

%TIM: I think we can just say this. Not sure it needs a lemma. This would save some space.
\begin{lem}\label{lem:simp-finite}
$\procsslra$ is finite for $\varphi[ \vec{ \con k }, \con e ]$.
\end{lem}
\begin{proof}
The terms returned by $\procsslra$ are in $\{ \ell_1, \ldots, \ell_n \}$ when $n>0$, and in $\{ u_1, \ldots, u_m \}$ when $n=0$.
\qed
\end{proof}

\begin{lem}\label{lem:simp-monotonic}
$\procsslra$ is monotonic for $\varphi[ \vec{ \con k }, \con e ]$.
\end{lem}
\begin{proof}
Let $\I$ be a model of $\lra$ and $\Gamma$, where $\Gamma \models \varphi[\vec{\con k}, t]$,
and assume by contradiction $\procsslra$ returns $t$.
Consider the case where $n>0$ and $t = \ell_i$ for some $i \in \{ 1, \ldots, n \}$.
Since $\I \models \Gamma \cup \lnot \varphi[ \vec{ \con k }, \con e ]$, it satisfies:
\begin{equation}
\begin{aligned}[c]
( \ell_i < \ell_1 \vee \ldots \vee \ell_i < \ell_n \vee \ell_i > u_1 \vee \ldots \vee \ell_i > u_m ) \wedge \\
\con e \geq \ell_1 \wedge \ldots \wedge \con e \geq \ell_n \wedge \con e \leq u_1 \wedge \ldots \wedge \con e \leq u_m
\end{aligned}
\end{equation}
Assume $\I$ satisfies $\ell_i > u_k$ for some $k \in \{ 1, \ldots, m \}$.
We have $\I$ also satisfies $\con e \geq \ell_i$ and $\con e \leq u_k$, and thus $\con e^\I \geq \ell_i^\I > u_k^\I \geq \con e^\I$.
Thus, $\I$ must satisfy $\ell_i < \ell_{k'}$ for some $k' \in \{ 1, \ldots, n \}$, and thus $\mathsf{max} \{ \ell_1^\I, \ldots, \ell_n^\I \} \neq \ell_i^\I$.
Thus, $\procsslra( \Gamma ) \neq \ell_i$.
By symmetrical reasoning when $n=0, m>0$ and $t = u_j$ for some $j \in \{ 1, \ldots, m \}$,
we have that $\procsslra( \Gamma ) \neq u_j$.
Thus, $\procsslra$ does not return $t$, and thus is monotonic for $\varphi[ \vec{ \con k }, \con e ]$.
\qed
\end{proof}

\begin{example}
Consider the formula $\forall x\, ( x < b \vee x > a )$.
The negated skolemized form ($\lnot \varphi[\vec{\con k}, \con e]$ in Figure~\ref{fig:proc-qi}) of this formula is equivalent to the formula $b \leq e \wedge e \leq a$ where $e$ is a fresh constant.
% In $\mathcal{P}_{\procsslra}$, we consider the negated skolemized form of this formula, denoted $\lnot \varphi[\vec{\con k}, \con e]$ in Figure~\ref{fig:proc-qi}, 
% which in this case is equivalent to the formula $b \leq e \wedge e \leq a$ where $e$ is a fresh constant.
A possible run of $\mathcal{P}_{\procsslra}$ on this input is as follows,
where $\Gamma$ is initially $\emptyset$
and on each iteration $\Gamma' = \Gamma \cup \{ b \leq e \wedge e \leq a \}$.
\[\begin{array}{c@{\hspace{1em}}c@{\hspace{1em}}c@{\hspace{1em}}c@{\hspace{1em}}c@{\hspace{1em}}l}
\hline
\text{\#} & \Gamma & \Gamma' &  & t[\vec{\con k}] & \text{Add to } \Gamma  \\
\hline
 1 & \text{sat} & \text{sat} & \mathsf{max} \{ b^\I \} = b^\I & b &  b < b \lor b > a
 \\
 2 & \text{sat} & \text{unsat} & & \\
\hline
\end{array}
\]
In step 2, note that $\Gamma' = \{ b < b \lor b > a, b \le e \land e \le a \}$, which is unsat.
The run establishes that $\exists ab\, \forall x\, ( x > a \vee x < b )$ is $\lra$-satisfiable.
\qed
\end{example}

%TIM: This example is not particularly informative. Maybe drop?
\begin{example}
Consider the formula $\forall x\, ( x < a \vee x < b )$,
whose skolemized negation is equivalent to $e \geq a \wedge e \geq b$.
A possible run of $\mathcal{P}_{\procsslra}$ on this input is as follows.
\[\begin{array}{c@{\hspace{1em}}c@{\hspace{1em}}c@{\hspace{1em}}c@{\hspace{1em}}c@{\hspace{1em}}l}
\hline
\text{\#} & \Gamma & \Gamma' &  & t[\vec{\con k}] & \text{Add to } \Gamma  \\
\hline
 1 & \text{sat} & \text{sat} & \mathsf{max} \{ a^\I, b^\I \} = a^\I & a & a < a \lor a < b 
 \\
 2 & \text{sat} & \text{sat} & \mathsf{max} \{ a^\I, b^\I \} = b^\I & b & b < a \lor b < b
 \\
 3 & \text{unsat} & & & \\
\hline
\end{array}
\]
Note that
the formula added in step 1, equivalent to $a < b$, 
ensures that, in the second iteration, $\mathsf{max} \{ a^\I, b^\I \} \neq a^\I$.
The run establishes that $\exists ab\, \forall x\, ( x > a \vee x > b )$ is $\lra$-unsatisfiable,
as expected in a linear order without endpoints.
\qed
\end{example}
\end{longv}

\section{Instantiation for Quantifier-Free $\lra$-Formulas}
\label{sec:lra}

Consider the case where $\vec a$ and $\vec x$ are vectors of $\Real$ variables
and $\lan$ is the class of formulas $\exists \vec a\, \forall \vec x\, \varphi[ \vec a, \vec x ]$
where $\varphi[ \vec a, \vec x ]$ is an arbitrary quantifier-free $\lra$-formula.
We assume that equalities are eliminated from $\varphi$ 
\begin{longv}
by the transformation:
\[
t \teq 0  \transform  0 \leq t \wedge 0 \geq t
\]
\end{longv}
\begin{shortv}
by replacing $t \teq s$ with $s \leq t \wedge s \geq t$.

\end{shortv}
Figure~\ref{fig:sel-lra} gives a selection function $\procslra$ for $\lra$,
which takes an interpretation $\I$, a set of formulas $\Gamma$, and the formula $\lnot \varphi[\vec{ \con k }, \vec{\con e}]$.
It invokes the recursive procedure $\procslrah$ which constructs a term 
corresponding to each variable in $\vec{ \con e }$.
Analogous to existing approaches for linear quantifier elimination~\cite{Loos93applyinglinear,nipkow2008linear},
our approach makes use of non-standard terms for symbolically representing substitutions.
In particular, the terms we consider may involve a free distinguished constant
$\delta$, representing an infinitesimal positive value.
For each variable $\con e_i$ from $\vec{ \con e }$, the procedure $\procslrah$ 
invokes the (non-deterministic) subprocedure $\procslrahh$, which chooses a term corresponding to $\con e_i$
based on a set of literals $M$ over the atoms of $\psi$ which propositionally entail $\psi$ and are satisfied by $\I$, 
% TIM: This is a really misleading name! Propositional assignment? Propositional abstraction?
% AJR: changed.
which we call a \emph{propositionally satisfying assignment} for $\psi$.
We partition $M$ into three sets $M_\ell$, $M_u$ and $M_c$, where 
$M_\ell$ contains literals that correspond to lower bounds for $\con e$,
$M_u$ contains literals that correspond to upper bounds for $\con e$,
and $M_c$ contains the remaining literals.
The sets $M_\ell$ and $M_u$ are equivalent to sets of literals that are in solved form with respect to $\con e$.
When $M_\ell$ contains at least one literal, we may return the lower bound whose value is maximal according to $\I$,
and similarly for $M_u$.
If both $M_\ell$ and $M_u$ are empty, we return the term $0$.
When $\procslrahh$ returns the term $t_i$, 
we apply the substitution $\{ \con e_i \mapsto t_i \}$ to $\psi$ and $\vec t$,
and append $t_i$ to $\vec{t}$.
Terms returned by $\procslrahh$ may involve the constant $\delta$.
We define a satisfiability relation between models and formulas involving $\delta$, as well as
the $\mathsf{max}$ and $\mathsf{min}$ function for terms involving $\delta$ in the obvious way,
such that
$( t_1 + c_1 \cdot \delta )^\I > ( t_2 + c_2 \cdot \delta )^\I$
if either $t_1^\I > t_2^\I$ or both $t_1^\I = t_2^\I$ and $c_1 > c_2$.

Overall, $\procslra$ returns a tuple of terms $\vec t$,
after which we add the instance $\normalize{\varphi[ \vec{\con k}, \vec t]}$ to $\Gamma$ in Figure~\ref{fig:proc-qi}.
\begin{longv}
We assume $\normalize{}$ eliminates occurrences of $\delta$ by the following transformations,
which is inspired by virtual term substitution~\cite{Loos93applyinglinear}:
% TIM: I would phrase this differently. I would say this is "related to" or is "inspired by" virtual term substition. "refered to as" sounds like a controversial claim.
\[\begin{array}{r@{}c@{}l@{}c@{}r@{}c@{}l@{}l}
\delta < t & \transform & 0 < t & \text{ and } &
\delta > t & \transform & 0 \geq t & \text{ where } \delta \not\in \fvars( t ).
\end{array}\]
\end{longv}
\begin{shortv}
We assume $\normalize{}$ eliminates occurrences of $\delta$ by
replacing literals containing $\delta$ whose solved form with respect to $\delta$ is of the form $(\lnot)(\delta \bowtie s)$ 
by $(\lnot)( 0 < s)$ if $\bowtie$ is $<$, and $(\lnot)( 0 \geq s)$ if $\bowtie$ is $>$.
This approach is inspired by virtual term substitution~\cite{Loos93applyinglinear}.
\end{shortv}

\begin{figure}[t]
\begin{framed}
$\procslra( \I, \Gamma, \lnot \varphi[\vec{ \con k }, \vec{\con e}] )$:
\begin{itemize}
\item[\ ] Return $\procslrah( \I, \lnot \varphi[\vec{ \con k }, \vec{\con e}], \vec{ \con e }, () )$
\end{itemize}

$\procslrah( \I, \psi, ( \con e_i, \ldots, \con e_n ), \vec{t} )$:
\begin{itemize}
\item[\ ] If $i>n$, return $\vec{t}$
\item[\ ] Otherwise, let $t_i = \procslrahh( \I, \psi, \con e_i )$, $\sigma = \{ \con e_i \mapsto t_i \}$
\item[\ ] Return $\procslrah( \I, \psi \sigma, ( \con e_{i+1}, \ldots, \con e_n ), ( \vec{t} \sigma, t_i ) )$
\end{itemize}

$\procslrahh( \I, \psi, \con e )$:
\begin{itemize}
\item[\ ] Let $M = M_\ell \cup M_u \cup M_c$ be such that:
  \begin{itemize}
  \item $\I \models M$ and $M \models_p \psi$,
  \item $M_\ell \Leftrightarrow \{ \con e \succ \ell_1, \ldots, \con e \succ \ell_n \}$,
  \item $M_u \Leftrightarrow \{ \con e \prec u_1, \ldots, \con e \prec u_m \}$, and
  \item $\con e \not\in \fvars( \ell_1, \ldots, \ell_n ) \cup \fvars( u_1, \ldots, u_m ) \cup \fvars( M_c )$.
  \end{itemize}
\item[\ ] Return one of
  $\begin{cases}
  \ell_i + \delta^\ell_i & n>0, \mathsf{max} \{ (\ell_1 + \delta^\ell_1 )^\I, \ldots, (\ell_n + \delta^\ell_n )^\I \} = (\ell_i + \delta^\ell_i)^\I \\
  u_j - \delta^u_j &  m>0, \mathsf{min} \{ (u_1 - \delta^u_1 )^\I , \ldots, ( u_m - \delta^u_m )^\I \} = (u_j - \delta^u_j)^\I \\
  0 & n=0 \text{ and } m=0
  \end{cases}$
\end{itemize}
\end{framed}
\vspace{-2ex}
\caption{A selection function $\procslra$ for arbitrary quantifier-free $\lra$-formula $\varphi[ \vec a, x ]$.
Each $\prec$ is either $<$ or $\leq$; $\delta^\ell_i$ is $\delta$ if the $i^{th}$ lower bound for $\con e$ is strict, and $0$ otherwise.
Similarly, each $\succ$ is either $>$ or $\geq$; $\delta^u_j$ is $\delta$ if the $j^{th}$ upper bound for $\con e$ is strict, and $0$ otherwise.
\label{fig:sel-lra}}
\end{figure}

%In the following, we show the selection function is finite and model-preserving for
%formulas $\varphi[ \vec{ \con k }, \vec{\con e} ]$.

\begin{lem}\label{lem:lra-finite}
$\procslra$ is finite for $\varphi[ \vec{ \con k }, \vec{\con e} ]$.
\end{lem}
\begin{longv}
\begin{proof}
We first show that only a finite number of terms can be returned by $\procslrahh( \I, (\lnot \varphi[\vec{\con k}, \vec{\con e}])\sigma, \con e_i )$
for any $\I, \sigma$.
Let $A$ be the set of atoms occurring in $\varphi[ \vec{ \con k }, \con e_i ] \sigma$.
The literals in satisfying assignments of $(\lnot \varphi[\vec{\con k}, \vec{\con e}] )\sigma$ are over these atoms.
Let $\{ \con e_i<\ell_1, \ldots, \con e_i<\ell_n, \con e_i>u_1, \ldots, \con e_i>u_m \}$ be the set of atoms
that are in solved form with respect to $\con e_i$ that are equivalent to the atoms of $A$ containing $\con e_i$,
where $\con e_i \not\in \fvars( \ell_1, \ldots, \ell_n, u_1, \ldots, u_m )$.
The terms returned by $\procslrahh( \I, ( \lnot \varphi[\vec{\con k}, \vec{\con e}] )\sigma, \con e_i )$ 
are in $\{ 0, \ell_1 (+ \delta ), \ldots, \ell_n (+ \delta ), u_1 (- \delta ), \ldots, u_m (- \delta ) \}$.
Since there are only a finite number of recursive calls to $\procslrah$ within $\procslra$,
and each call appends only a finite number of possible terms to $\vec t$,
the set of possible return values of $\procslra$ is finite, and thus it is finite for $\varphi[ \vec{ \con k }, \vec{\con e} ]$.
\qed
\end{proof}
\end{longv}

\begin{lem} \label{lem:procslrahh-model-preserve}
If $\I$ is a model for $\lra$ and quantifier-free formula $\psi$, 
then $\I$ is also a model for $\psi \{ \con e \mapsto \procslrahh( \I, \psi, \con e ) \}$.
\end{lem}
\begin{longv}
\begin{proof}
Let $M$ be a set of literals of the form described in the definition of $\procslrahh$ for $\I$, $\psi$ and $\con e$.
Consider the case where $\procslrahh( \I, \psi, \con e ) = \ell_i + \delta^\ell_i$ for some $i$, where $n>0$.
We show that $\I$ satisfies $M \{ \con e \mapsto \ell_i + \delta^\ell_i \}$.
% TIM: Maybe say this in words?
% First, since  (\ell_i + \delta^\ell_i)^\I is the maximum lower bound w.r.t. \I,
First, since $\mathsf{max} \{ (\ell_1 + \delta^\ell_1)^\I , \ldots, (\ell_n + \delta^\ell_n)^\I \} = (\ell_i + \delta^\ell_i)^\I$,
we know that $\I$ satisfies $M_\ell \{ \con e \mapsto \ell_i + \delta^\ell_i \}$.
% TIM: The proof is missing the claim that l_i <= e <= u_j w.r.t. \I.
In the case that the bound on $\con e$ we consider is strict, that is,
$\con e > \ell_i \in M_\ell$, then $\delta^\ell_i$ is $\delta$, and $\ell_i^\I < u_j^\I$ for all $j \in \{ 1, \ldots, m \}$.
Thus, $\I$ satisfies $(\ell_i + \delta \prec u_j) = (\con e \prec u_j) \{ \con e \mapsto \ell_i + \delta \}$.
In the case that the bound on $\con e$ we consider is non-strict, that is,
if $\con e \geq \ell_i \in M_\ell$, then $\delta^\ell_i$ is $0$, and $\ell_i^\I \leq u_j^\I$ for all $j \in \{ 1, \ldots, m \}$.
Thus, $\I$ satisfies $(\ell_i \prec u_j) = (\con e \prec u_j) \{ \con e \mapsto \ell_i \}$.
In either case, we have that $\I$ satisfies each literal in $M_u \{ \con e \mapsto \ell_i + \delta^\ell_i \}$.
Finally, $\I$ clearly satisfies $M_c \{ \con e \mapsto \ell_i + \delta^\ell_i \} = M_c$.
The case when $m>0$ is symmetric to the case when $n>0$.
In the case where $n=0$ and $m=0$, 
we have that $\psi$ does not contain $\con e$, and $\I$ satisfies $M \{ \con e \mapsto 0 \}$.
In each case, $\I$ satisfies $M \{ \con e \mapsto \procslrahh( \I, \psi, \con e ) \}$,
which entails $\psi \{ \con e \mapsto \procslrahh( \I, \psi, \con e ) \}$, and thus the lemma holds.
\qed
\end{proof}
\end{longv}

% TIM: The next two proofs can be turned into sentences.
\begin{lem}\label{lem:lra-monotonic}
$\procslra$ is model-preserving for $\varphi[ \vec{ \con k }, \vec{\con e} ]$.
\end{lem}
\begin{proof}
By the definition of $\procslrah$ and repeated applications of Lemma~\ref{lem:procslrahh-model-preserve}.
%Assume by contradiction $\procslra( \I, \Gamma, \lnot \varphi[\vec{ \con k }, \vec{\con e}] ) = \vec{t}$ where $\varphi[\vec{\con k}, \vec{t}] \in \Gamma$.
%By definition of selection function, $\I$ is a model of $\lra$, $\Gamma$ and $\lnot \varphi[\vec{ \con k }, \vec{\con e}]$.
%By the definition of $\procslrah$, and repeated applications of Lemma~\ref{lem:procslrahh-model-preserve},
%we have that $\I$ satisfies $( \lnot \varphi[\vec{ \con k }, \vec{\con e}] ) \{ \vec{ \con e } \mapsto \vec{t} \} =$ 
%$\lnot \varphi[\vec{ \con k }, \vec{t}] $, which is a contradiction since $\I$ satisfies $\varphi[\vec{\con k}, \vec{t}] \in \Gamma$.
%Thus, $\procslra( \I, \Gamma, \vec{\con e} ) \neq \vec{t}$ and $\procslra$ is monotonic for $\varphi[ \vec{ \con k }, \vec{\con e} ]$.
\qed
\end{proof}

\begin{thm}
$\mathcal{P}_{\procslra}$ is a sound and complete procedure 
for determining the $\lra$-satisfiability of $\exists \vec a\, \forall \vec x\, \varphi[\vec a, \vec x]$.
\end{thm}
\begin{proof}
By Theorem~\ref{lem:sel-fin-mono} and Lemma~\ref{lem:sel-mm-m} of our framework as well as $\lra$-specific 
Lemma~\ref{lem:lra-finite} and Lemma~\ref{lem:lra-monotonic}.
\qed
\end{proof}

\begin{shortv}
\ 

\noindent The proof of Lemma~\ref{lem:lra-finite} relies on the fact that only a fixed finite set of possible terms can be returned by $\procslrahh$, 
regardless of $\I$\footnote{For full details, see~\cite{RKK2015}.}. Unlike test terms in conventional quantifier elimination, the terms in our approach are generated 
one by one and only when
needed.
Lemma~\ref{lem:procslrahh-model-preserve} is the crux of the correctness of the approach: it gives
semantic argument for progress in the finite but large space of instantiations.
The proof relies on the fact that 
if $\I$ interprets $\con e$ such that $\con e^\I$ satisfies all lower bounds in $M_\ell$ (resp. upper bounds in $M_u$), 
we may replace $\con e^\I$ with the minimal (resp. maximal) value among the bounds for $\con e$ according to $\I$.
The lemma is not a trivial fact about instantiation of quantifiers;
$\psi$ is not quantified. The term returned from
$\procslra$ ``generalizes'' the value of $\con e$ in a $\I$ to an expression in terms of 
other variables, while keeping the formula true in $\I$. Figure~\ref{fig:proc-qi} then adds a formula instance that 
blocks the model $\I$ (and others).
\end{shortv}

% TIM: Sentence does not add much.
\begin{shortv}
We illustrate the procedure through examples\footnote{A larger set of examples can be found in the extended version of this report~\cite{RKK2015}.}.
\end{shortv}
\begin{longv}
We illustrate the procedure through examples.
\end{longv}
% TIM: Make the non-determinism clearer earlier. 
$\procslrahh$ is non-deterministic;
we choose instantiations only based on the lower bounds $M_\ell$ found in the procedure $\procslrahh$,
though the procedure is free to base its instantiations on the upper bounds $M_u$ as well.
We underline the literal in $M_\ell$ corresponding to the bound whose value is maximal in $\I$.
$\Gamma$ is initially empty and on each iteration $\Gamma'$ is the union of $\Gamma$ and the skolemized negation of the input formula.
Each round of $\procslra$ computes a tuple $\vec t[\vec{\con k}]$, which is used to instantiate our quantified formula in Figure~\ref{fig:proc-qi}.
The last column shows the corresponding instance of the quantified formula after simplification,
including the elimination of $\delta$.

\begin{shortv}
\begin{example}
Consider the formula $\forall x\, ( x < b \vee x > a )$.
The negated skolemized form (denoted $\lnot \varphi[\vec{\con k}, \con e]$ in Figure~\ref{fig:proc-qi}) of this formula 
is equivalent to the formula $e \geq b \wedge e \leq a$, where $e$ is a fresh constant.
A possible run of $\mathcal{P}_{\procslra}$ on this input is as follows.
%where $\Gamma$ is initially $\emptyset$
%and on each iteration $\Gamma' = \Gamma \cup \{ b \leq e \wedge e \leq a \}$.
\[\begin{array}{c@{\hspace{1em}}c@{\hspace{1em}}c|@{\hspace{1em}}c@{\hspace{1em}}c@{\hspace{1em}}c@{\hspace{1em}}|c@{\hspace{1em}}l}
\hline
 &  &  & \multicolumn{3}{c|}{\procslrahh( \I, \Gamma, \con e )} &   \\
\text{\#} & \Gamma & \Gamma' & \con{e} & M_\ell & \text{return} & \vec t[\vec{\con k}] & \text{Add to } \Gamma  \\
\hline
 1 & \text{sat} & \text{sat} & e & \{ \underline{e \geq b} \} & b & (b) & b > a
 \\
 2 & \text{sat} & \text{unsat} & & & & \\
\hline
\end{array}
\]
In step 1, we add to $\Gamma$ the instance $b < b \lor b > a$, which simplifies to $b > a$.
In step 2, note that $\Gamma' = \{ b > a, b \le e \land e \le a \}$, which is unsatisfiable.
The run establishes that $\exists ab\, \forall x\, ( x > a \vee x < b )$ is $\lra$-satisfiable.
\qed
\end{example}
\end{shortv}

% TIM: Definitely do not want both this and example 1.
\begin{example}
To demonstrate how non-strict bounds are handled,
consider the formula $\forall x\, ( x \leq a \vee x \leq b )$,
whose skolemized negation is $e > a \wedge e > b$.
A possible run of $\mathcal{P}_{\procslra}$ on this input is as follows.
\[\begin{array}{c@{\hspace{1em}}c@{\hspace{1em}}c|@{\hspace{1em}}c@{\hspace{1em}}c@{\hspace{1em}}c@{\hspace{1em}}|c@{\hspace{1em}}l}
\hline
 &  &  & \multicolumn{3}{c|}{\procslrahh( \I, \Gamma, \con e )} &   \\
\text{\#} & \Gamma & \Gamma' & \con{e} & M_\ell & \text{return} & \vec t[\vec{\con k}] & \text{Add to } \Gamma  \\
\hline
 1 & \text{sat} & \text{sat} & e & \{ \underline{e > a}, e > b \} & a + \delta & (a+\delta) & a < b
 \\
 2 & \text{sat} & \text{sat} & e & \{ e > a, \underline{e>b} \} & b + \delta & (b+\delta) & b < a
 \\
 3 & \text{unsat} & & & & & \\
\hline
\end{array}
\]
This run shows $\exists ab\, \forall x\, ( a \geq x \vee x \geq b )$ is $\lra$-unsatisfiable.
The disjuncts of the instance $a+\delta \leq a \vee a+\delta \leq b$ added to $\Gamma$ on the first iteration
simplify to $\bot$ and $a < b$ respectively.  We similarly obtain $b < a$ on the second iteration.
% TIM: QED at the end of an example?
\qed
\end{example}

\begin{example}
To demonstrate how multiple universally quantified variables are handled,
consider the formula $\forall xy\, ( x+y < a \vee x-y < b )$
whose skolemized negation is $e_1+e_2 \geq a \wedge e_1-e_2 \geq b$.
A possible run of $\mathcal{P}_{\procslra}$ %on this input 
is as follows.
\[\begin{array}{c@{\hspace{1em}}c@{\hspace{1em}}c|@{\hspace{1em}}c@{\hspace{1em}}c@{\hspace{1em}}c@{\hspace{1em}}|c@{\hspace{1em}}l}
\hline
 &  &  & \multicolumn{3}{c|}{\procslrahh( \I, \Gamma, \con e )} &   \\
\text{\#} & \Gamma & \Gamma' & \con{e} & M_\ell  & \text{return} & \vec t[\vec{\con k}] & \text{Add to } \Gamma  \\
\hline
 1 & \text{sat} & \text{sat} & 
 e_1 & \{ e_1 \geq a - e_2, \underline{e_1 \geq b + e_2} \} & b + e_2 & 
 \\
  & & & 
 e_2 & \{ \underline{e_2 \geq \frac{a-b}{2}} \} & \frac{a-b}{2} & ( \frac{a+b}{2},\frac{a-b}{2} ) & \bot
 \\
 2 & \text{unsat} & & & & \\
\hline
\end{array}
\]
This run shows $\exists ab\, \forall xy\, ( x+y < a \vee x-y < b )$ is $\lra$-unsatisfiable.
The substitution for $e_2$ is chosen based on $M_\ell$ after
applying the substitution $\{ e_1 \mapsto b + e_2 \}$.
\qed
\end{example}

\begin{longv}
% TIM: This is a really clever example, but maybe just refer readers to the longer version?
% I am not sure more casual readers [or reviewers] will follow.
\begin{example}
To demonstrate how quantified formulas with Boolean structure are handled,
consider the formula $\forall x\, ( ( a < x \wedge x < b ) \vee x < a+b )$
whose skolemized negation is $( a \geq e \vee e \geq b ) \wedge e \geq a+b$.
A possible run of $\mathcal{P}_{\procslra}$ %on this input
is as follows.
\[\begin{array}{c@{\hspace{1em}}c@{\hspace{1em}}c|@{\hspace{1em}}c@{\hspace{1em}}c@{\hspace{1em}}c@{\hspace{1em}}|c@{\hspace{1em}}l}
\hline
 &  &  & \multicolumn{3}{c|}{\procslrahh( \I, \Gamma, \con e )} &   \\
\text{\#} & \Gamma & \Gamma' & \con{e} & M_\ell  & \text{return} & \vec t[\vec{\con k}] & \text{Add to } \Gamma  \\
\hline
 1 & \text{sat} & \text{sat} & 
 e & \{ \underline{e \geq a+b} \} & a+b & (a+b) & 0 < b \wedge a < 0
 \\
 2 & \text{sat} & \text{sat} & 
 e & \{ \underline{ e \geq b}, e \geq a+b \} & b & (b) & 0 < a
 \\
 3 & \text{unsat} & & & &
 \\
\hline
\end{array}
\]
This run shows $\exists ab\, \forall x\, ( ( a < x \wedge x < b ) \vee x < a+b )$ is $\lra$-unsatisfiable.
On the first iteration, 
we assume that the propostionally satisfying assignment for $\Gamma'$ included $a \geq e$,
and hence $e \geq b$ is not included as a lower bound on that iteration.
On the second iteration, 
the solver must satisfy both $b > 0$ and $a < 0$, 
which implies the model $\I$ is such that
$(a+b)^{\I} > a^{\I}$ hence $e \geq b$ must exist in $M_\ell$,
and moreover $b^{\I} > (a+b)^{\I}$ hence $b$ must be the maximal lower bound for $e$.
\qed
\end{example}

% TIM: This can definitely be dropped!
\begin{example}
To demonstrate a case where a variable has no bounds,
consider the formula $\forall xy\, x \leq y$,
whose skolemized negation is $e_1 > e_2$.
A possible run of $\mathcal{P}_{\procslra}$ on this input is as follows.
\[\begin{array}{c@{\hspace{1em}}c@{\hspace{1em}}c|@{\hspace{1em}}c@{\hspace{1em}}c@{\hspace{1em}}c@{\hspace{1em}}|c@{\hspace{1em}}l}
\hline
 &  &  & \multicolumn{3}{c|}{\procslrahh( \I, \Gamma, \con e )} &   \\
\text{\#} & \Gamma & \Gamma' & \con{e} & M_\ell  & \text{return} & \vec t[\vec{\con k}] & \text{Add to } \Gamma  \\
\hline
 1 & \text{sat} & \text{sat} & 
 e_1 & \{ \underline{e_1 > e_2} \} & e_2 + \delta & 
 \\
  & & & 
 e_2 & \emptyset & 0 & ( \delta, 0 ) & \bot
 \\
 2 & \text{unsat} & & & &
 \\
\hline
\end{array}
\]
This run shows $\forall xy\, x > y$ is $\lra$-unsatisfiable.
Notice that after the substitution $\{ e_1 \mapsto e_2 + \delta \}$, 
we have that $\Gamma'$ contains neither an upper nor a lower bound for $e_2$, and hence we choose to return the value $0$.
\qed
\end{example}

\begin{example}
To demonstrate a non-trivial case using the infinitesimal $\delta$,
consider the formula $\forall xy\, (x \leq 0 \vee y - 2 \cdot x \leq 0)$
whose skolemized negation is $e_1  > 0 \wedge e_2 - 2 \cdot e_1 > 0$.
A possible run of $\mathcal{P}_{\procslra}$ on this input is as follows.
\[\begin{array}{c@{\hspace{1em}}c@{\hspace{1em}}c|@{\hspace{1em}}c@{\hspace{1em}}c@{\hspace{1em}}c@{\hspace{1em}}|c@{\hspace{1em}}l}
\hline
 &  &  & \multicolumn{3}{c|}{\procslrahh( \I, \Gamma, \con e )} &   \\
\text{\#} & \Gamma & \Gamma' & \con{e} & M_\ell  & \text{return} & \vec t[\vec{\con k}] & \text{Add to } \Gamma  \\
\hline
 1 & \text{sat} & \text{sat} & 
 e_1 & \{ \underline{e_1 > 0} \} & \delta & 
 \\
  & & & 
 e_2 & \{ \underline{e_2 > 2 \cdot \delta} \} & 3 \cdot \delta & ( \delta, 3 \cdot \delta ) & \bot
 \\
 2 & \text{unsat} & & & &
 \\
\hline
\end{array}
\]
This run shows $\forall xy\, x \leq 0 \vee y - 2 \cdot x \leq 0$ is $\lra$-unsatisfiable.
\qed
\end{example}
\end{longv}

The procedure $\mathcal{P}_{\procslra}$, which is an instance of the procedure in Figure~\ref{fig:proc-qi},
can be understood as lazily enumerating the disjuncts of the Loos-Weispfenning method 
for quantifier elimination over linear real arithmetic~\cite{Loos93applyinglinear},
with minor differences\footnote{
For instance, that method uses a distinguished term $\infty$ representing an arbitrarily large positive value.
\begin{shortv}
A more technical and thorough comparison of how our approaches for \lra and \lia relate to existing approaches can be found in~\cite{RKK2015}.
\end{shortv}
}.
In this way, our approach is similar to the projection-based procedures described in~\cite{komuravelli2014smt,bjornerplaying}.
These approaches compute implicants of quantified formulas, while our approach instead computes a term which is in turn used for instantiation.
This choice is important for our purposes, as 
it enables a uniform combination of the approach with existing instantiation-based techniques for first-order logic~\cite{ganzinger2003new,Detlefs03simplify:a,reynolds14quant_fmcad},
and allows the approach to be used as a subprocedure for synthesis as described in~\cite{ReynoldsDKBT15Cav}.

\begin{figure}[t]
\begin{framed}
Return
  $\begin{cases}
  \frac{u_j + \ell_i}{2} & n>0 \text{ and } m>0 \\
  \ell_i + 1 &  n>0 \text{ and } m=0  \\
  u_j - 1 & n=0 \text{ and } m>0 \\
  0 & n=0 \text{ and } m=0
  \end{cases}$, where 
  $\begin{cases}
%\mathsf{max} \{ (\ell_k + \delta^\ell_k )^\I \}_{k=1}^n = (\ell_i + \delta^\ell_i)^\I \text{ if } n>0 \\
%\mathsf{min} \{ (u_k - \delta^u_k )^\I \}_{k=1}^m = (u_j - \delta^u_j)^\I \text{ if } m>0.
\mathsf{max} \{ \ell_1^\I, \ldots, \ell_n^\I \} = \ell_i^\I \text{ if } n>0 \\
\mathsf{min} \{ u_1^\I , \ldots, u_m^\I \} = u_j^\I \text{ if } m>0.
\end{cases}$.
\end{framed}
\vspace{-2ex}
\caption{An alternative return value for $\procslrahh$.
\label{fig:sel-lra2}}
\end{figure}

An alternative return value for $\procslrahh$ is presented in Figure~\ref{fig:sel-lra2}.
When using this return value, the procedure $\mathcal{P}_{\procslra}$
enumerates the disjuncts of Ferrante and Rackoff's method for quantifier elimination over linear real arithmetic~\cite{FerranteRackoff79ComputationalComplexityLogicalTheories},
with minor differences.
We provide an experimental evaluation of both of these selection functions in Section~\ref{sec:results}.

\begin{longv}
\subsection{Comparison to Existing Approaches}

As mentioned, at their core, most approaches for solving quantified linear arithmetic (including ours) 
share many similarities with one another.
In particular, given an existentially quantified formula $\exists \vec x. \varphi$, based on some strategy, 
they enumerate (possibily lazily) a finite set of ground formulas that are entailed by this formula.
We give a brief overview contrasting the technical details of existing approaches in this section.

We have mentioned that the approach in Figure~\ref{fig:sel-lra} involves the use of a free distinguished constant $\delta$,
representating an infinitessmal positive value.
Other approaches also involve use of a free distinguished constant $\infty$, representing an arbitrarily large positive value.
Like $\delta$, this term can be eliminated, as follows:
\[\begin{array}{r@{}c@{}l@{}c@{}r@{}c@{}l@{}l}
\infty < t & \transform & \bot & \text{ and } &
\infty > t & \transform & \top & \text{ where } \infty \not\in \fvars( t ).
\end{array}\]

\begin{figure}[t]
\begin{framed}
Return one of
  $\begin{cases}
  l_j + \delta & n>0 \\
  u_j - \delta & m>0 \\
  \infty & m=0  \\
  -\infty & n=0
  \end{cases}$, where 
  $\begin{cases}
\mathsf{max} \{ \ell_1^\I, \ldots, \ell_n^\I \} = \ell_i^\I \text{ if } n>0 \\
\mathsf{min} \{ u_1^\I , \ldots, u_m^\I \} = u_j^\I \text{ if } m>0.
\end{cases}$.
\end{framed}
\vspace{-2ex}
\caption{An alternative return value for $\procslrahh$ that is analogous to Loos and Weispfenning's method.
\label{fig:sel-lra3}}
\end{figure}

\begin{figure}[t]
\begin{framed}
Return
  $\begin{cases}
  \frac{u_j + \ell_i}{2} & n>0 \text{ and } m>0 \\
  \infty & m=0 \\
  -\infty & n=0
  \end{cases}$, where 
  $\begin{cases}
\mathsf{max} \{ \ell_1^\I, \ldots, \ell_n^\I \} = \ell_i^\I \text{ if } n>0 \\
\mathsf{min} \{ u_1^\I , \ldots, u_m^\I \} = u_j^\I \text{ if } m>0.
\end{cases}$.
\end{framed}
\vspace{-2ex}
\caption{An alternative return value for $\procslrahh$ that is analogous to Ferrante and Rackoff's method.
\label{fig:sel-lra4}}
\end{figure}

The two most widely known algorithms for quantifier elimination for linear real arithmetic 
are the method based on infinitessimals in~\cite{Loos93applyinglinear},
and the method based on an interior point method in~\cite{FerranteRackoff79ComputationalComplexityLogicalTheories}.
To put these algorithms into the context of our approach,
we provide two additional alternatives for the return value of $\procslrahh$ (Figures~\ref{fig:sel-lra3} and~\ref{fig:sel-lra4})
that closely approximate the effect of these methods.

Recent approaches are inspired by of one (or both) of these methods.
The approaches described in~\cite{komuravelli2014smt,bjornerplaying} are closely based on the Loos-Weispfenning method,
and the approach described in~\cite{dutertresolving} is closely based on Ferrante-Rackoff method.
The approach described in~\cite{nipkow2008linear} examines a certified version of both approaches.

As mentioned, Figure~\ref{fig:sel-lra} is inspired by the Loos-Weispfenning method, but does not use infinities.
Similarly, the return value in Figure~\ref{fig:sel-lra2} is inspired by Ferrante-Rackoff method, but does not use infinities.
One possible advantage of the approach in Figure~\ref{fig:sel-lra2} is in the context of quantifier alternation.
In particular, that selection function does not use any virtual terms,
and thus we may consider instances of quantified formulas having nested quantification where eliminating virtual terms is not obvious.
A complete strategy for quantifier alternation using this selection function is the subject of future work.

\end{longv}

\section{Instantiation for Quantifier-Free $\lia$-Formulas}
\label{sec:lia}
% TIM: This is a bit of an awkward introductory sentence.
% How about, "We now turn our attention to the case..."
% AJR: agreed
We now turn our attention to the class of arbitrary $\lia$-formulas $\exists \vec a\, \forall x\, \varphi[ \vec a, \vec x ]$,
$\vec x$ and $\vec a$ are vectors of $\Int$ variables, and where $\varphi[ \vec a, \vec x ]$ is quantifier-free.
We again assume all equalities are eliminated from $\varphi$ by replacing them with a conjunction of inequalities.

Figure~\ref{fig:sel-lia} gives a selection function $\procslia$ for $\lia$.
%Similar to the procedure from the previous section,
%$\procslia$ depends on a model $\I$ of $\Gamma$ and $\lnot \varphi[\vec{ \con k }, \vec{\con e}]$ which guides its selection of terms.
%The procedure in this section uses \emph{substitutions with coefficients}, which are mappings of the form 
%$\sigma = \{ c \cdot x \mapsto t \}$ where $c$ is a positive integer, $x$ is an integer variable, and $t$ is an integer term.
%We may apply $\sigma$ to integer terms of the form $c \cdot ( d \cdot x + s )$ where $x \not\in \fvars( s )$,
%where $( c \cdot ( d \cdot x + s ) ) \sigma$ is defined as $d \cdot t + c \cdot s$.
%We define $( s_1 \bowtie s_2 )\sigma$ as $( c \cdot s_1 ) \sigma \bowtie ( c \cdot s_2 ) \sigma$ for $\bowtie \in \{ <, > \}$, and thus we can apply $\sigma$ to arbitrary $\lia$-formulas.
The procedure invokes the recursive procedure $\procsliah$,
which takes as arguments $\I$, $\lnot \varphi[\vec{ \con k },\vec{\con e}]$, 
variables $\vec{ \con e }$ that we have yet to incorporate into the substitutions,
an integer $\theta$,
terms $\vec{ t }$ found as substitutions for variables from $\vec{ \con e }$ so far,
and a tuple of symbols $\vec{p}$ from $\{ +, - \}$ which we refer to as \emph{polarities}.
% TIM: Just saying $\theta$ is initially 1 is a bit too mysterious. We need some context.
The role of $\theta$ will be to capture divisibility relationships through the procedure,
where $\theta$ is initially $1$.
The procedure invokes a call to $\procsliahh( \I, \psi, \con e_i )$ which based on the propositionally satisfying assignment for $\psi$
returns a tuple of the form $(c, t_i, p_i)$, where $c$ is a constant, $t_i$ is a term, and $p_i$ is a polarity.
The procedure for constructing the term $t_i$ in the procedure $\procsliahh$ is similar 
to the procedure $\procslrahh$ in the previous section,
where we find the lower bound of the form $c_i \cdot \con e \geq \ell_i$ such that the (rational) value $( \frac{\ell_i}{ c_i } )^\I$ is maximal, and similarly for $M_u$.
Additionally, $\procslrahh$ adds a constant $\rho$ to the maximal lower bound (resp. minimal lower bound).
This constant ensures that the returned term $t_i$ and $\con e$ are congruent modulo $\theta \cdot c$ in $\I$,
a fact which in part suffices to show the overall function to be model-preserving.
%as we argue in the proof of Lemma~\ref{lem:procsliah-subs}.
It then constructs a \emph{substitution with coefficients} $\sigma$ of the form $\{ c \cdot \con e_i \mapsto t_i \}$.
A substitution of this form may be applied to integer terms of the form $c \cdot ( d \cdot \con e_i + s )$ where $\con e_i \not\in \fvars( s )$,
where $( c \cdot ( d \cdot \con e_i + s ) ) \sigma$ is defined as $d \cdot t_i + c \cdot s$.
Additionally, we define $( s_1 \bowtie s_2 )\sigma$ as $( c \cdot s_1 ) \sigma \bowtie ( c \cdot s_2 ) \sigma$ for $\bowtie \in \{ <, > \}$, 
and thus we can apply $\sigma$ to arbitrary $\lia$-formulas.
After constructing $\sigma$, the procedure $\procsliah$ invokes a recursive call where
$\sigma$ is applied to $\psi$ and $(c \cdot \vec{t})$,
$\theta$ is multiplied by $c$,
the term $\theta \cdot t_i$ is appended to $\vec{t}$,
and $p_i$ is appended to $\vec{p}$.

Overall, $\procsliah$ returns a vector of terms $(\vec t\ \opdivd^{\vec p} \ \theta)$, that is,
integer division applied pairwise to the terms in $\vec{t}$ and the constant $\theta$,
where $\vec{p}$ determines whether this division rounds up or down.
We add the instance $\normalize{\varphi[ \vec{\con k}, \vec t\ \opdivd^{\vec p} \ \theta]}$ to $\Gamma$ in Figure~\ref{fig:proc-qi},
where occurrences of integer division are eliminated by defining
$\normalize{\varphi[ \vec{\con k}, \vec t\ \opdivd^{\vec p} \ \theta]}$ as 
$\varphi[ \vec{\con k}, \vec{\con d}] \wedge \theta \cdot \vec{\con d} \teq \vec{t} \pm^{\vec p} \vec{\con m} \wedge 0 \leq \vec{\con m} < \theta$,
where $\vec{\con d}$ and $\vec{\con m}$ are distinct fresh constants, and $\pm^p$ is $+$ if $p$ is $+$ and analogously for $-$.
Note that our selection function chooses $\vec{p}$ such that integer division rounds up for terms coming from lower bounds,
and rounds down for terms combing from upper bounds.
\begin{longv}
This choice is not required for correctness, 
but can reduce the number of instances needed for showing unsatisfiability.
\end{longv}

%where $\vec{\con d}$ and $\vec{\con m}$ are tuples of fresh constants and $\oplus$ is a tuple of operators from $\{ +, - \}$.

%While not required for correctness,
%we assume $t_i \teq \theta \cdot d_i - m_i$ is chosen for
%terms coming from lower bounds in $\procsliahh$,
%and conversely $t_i \teq \theta \cdot d_i + m_i$ for upper bounds.
%In other words, $\normalize{}$ interprets integer division as rounding up for terms coming from lower bounds,
%and rounding down for terms coming from upper bounds,
%which we will see in the examples.

%While not required for correctness,
%we assume that 
%integer division rounding up is applied in this way to terms coming from lower bounds in $\procsliahh$, 
%which we will write as $\ceil{\frac{t_i}{\theta}}$,
%and integer division rounding down is applied to the rest.

\begin{figure}[t]
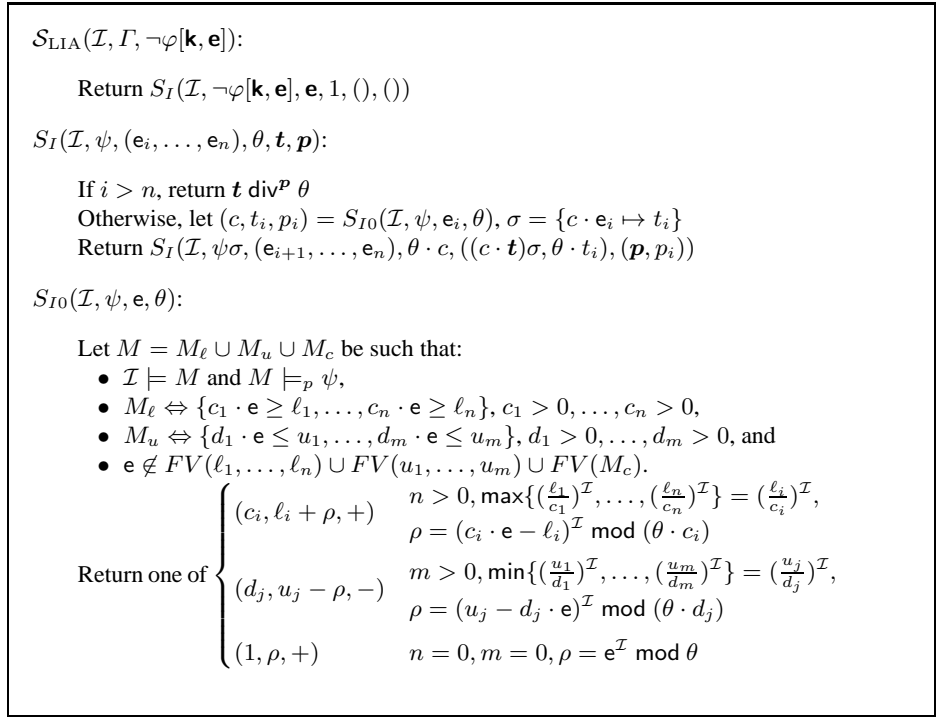

\begin{framed}
$\procslia( \I, \Gamma, \lnot \varphi[\vec{ \con k }, \vec{\con e}] )$:
\begin{itemize}
\item[\ ] Return $\procsliah( \I, \lnot \varphi[\vec{ \con k }, \vec{\con e}], \vec{ \con e }, 1, (), () )$
\end{itemize}

$\procsliah( \I, \psi, ( \con e_i, \ldots, \con e_n ), \theta, \vec t, \vec p )$:
\begin{itemize}
\item[\ ] If $i>n$, return $\vec t\ \opdivd^{\vec p} \ \theta$
\item[\ ] Otherwise, let $( c, t_i, p_i) = \procsliahh( \I, \psi, \con e_i, \theta )$, $\sigma = \{ c \cdot \con e_i \mapsto t_i \}$
\item[\ ] Return $\procsliah( \I, \psi \sigma, ( \con e_{i+1}, \ldots, \con e_n ), \theta \cdot c, ( ( c \cdot \vec{t} ) \sigma, \theta \cdot t_i ), (\vec p, p_i ) )$
\end{itemize}

$\procsliahh( \I, \psi, \con e, \theta )$:
\begin{itemize}
\item[\ ] Let $M =  M_\ell \cup M_u \cup M_c$ be such that:
  \begin{itemize}
  \item $\I \models M$ and $M \models_p \psi$,
  \item $M_\ell \Leftrightarrow \{ c_1 \cdot \con e \geq \ell_1, \ldots, c_n \cdot \con e \geq \ell_n \}$, $c_1 > 0, \ldots, c_n > 0$,
  \item $M_u \Leftrightarrow \{ d_1 \cdot \con e \leq u_1, \ldots, d_m \cdot \con e \leq u_m \}$, $d_1 > 0, \ldots, d_m > 0$, and
  \item $\con e \not\in \fvars( \ell_1, \ldots, \ell_n ) \cup \fvars( u_1, \ldots, u_m ) \cup \fvars( M_c )$.
  \end{itemize}
\item[\ ] Return one of
  $\begin{cases}
  (c_i, \ell_i + \rho, + ) & \parbox{1.0\textwidth}{$n>0, \mathsf{max} \{ ( \frac{\ell_1}{ c_1} )^\I, \ldots, ( \frac{\ell_n}{ c_n} )^\I \} = ( \frac{\ell_i}{ c_i } )^\I, \\
                                                 \rho = ( c_i \cdot \con e - \ell_i )^\I \ \opmod \ ( \theta \cdot c_i ) $}\\[1.2em]
  (d_j, u_j - \rho, - )    & \parbox{1.0\textwidth}{$m>0, \mathsf{min} \{ ( \frac{u_1}{ d_1 } )^\I, \ldots, ( \frac{u_m}{ d_m } )^\I \} = ( \frac{u_j}{ d_j } )^\I, \\
                                                 \rho = ( u_j - d_j \cdot \con e )^\I \ \opmod \ ( \theta \cdot d_j )$}\\[1.2em]
  (1, \rho, +) & n=0, m=0, \rho =  \con e ^\I \ \opmod \ \theta
  \end{cases}$
\end{itemize}
\end{framed}
\vspace{-2ex}
\caption{A selection function $\procslia$ for arbitrary quantifier-free $\lia$-formula $\varphi[ \vec a, x ]$.
%where $\opdivd$ is a function symbol denoting integer division.
\label{fig:sel-lia}}
\end{figure}

\begin{lem}\label{lem:lia-finite}
$\procslia$ is finite for $\varphi[ \vec{ \con k }, \vec{ \con e } ]$.
\end{lem}
\begin{longv}
\begin{proof}
First, we show that only a finite number of tuples are returned by $\procsliahh( \I, ( \lnot \varphi[\vec{\con k}, \vec{\con e}] )\sigma, \con e_i, \theta )$
for any $\I, \sigma, \con e_i$ and finite $\theta$.
Let $A$ be the set of atoms occurring in $\varphi[ \vec{ \con k }, \con e_i ] \sigma$.
The literals in satisfying assignments of $(\lnot \varphi[\vec{\con k}, \vec{\con e}] )\sigma$ are over these atoms.
Let $\{ c_1 \cdot \con e_i < \ell_1, \ldots, c_n \cdot \con e_i < \ell_n, d_1 \cdot \con e_i > u_1, \ldots, d_m \cdot \con e_i > u_m \}$ be the set of atoms
that are in solved form with respect to $\con e_i$ that are equivalent to the atoms of $A$ containing $\con e_i$,
where $\con e_i \not\in \fvars( \ell_1, \ldots, \ell_n, u_1, \ldots, u_m )$ and $c_1 > 0, \ldots, c_n > 0, d_1 > 0, \ldots, d_m > 0$.
The tuples returned by $\procsliahh( \I, ( \lnot \varphi[\vec{\con k}, \vec{\con e}] )\sigma, \con e_i )$ are in the finite set:
% TIM: Readers may have trouble understanding that you are also handling negations of atoms in the second column.
\[\begin{array}{c@{\hspace{1em}}c@{\hspace{1em}}c@{\hspace{1em}}c@{\hspace{1em}}c}
 \{ ( c_i, \ell_i + \rho, + ) \mid 1 \leq i \leq n \} & \cup & \{ ( d_j, u_j+1+\rho, - ) \mid 1 \leq j \leq m \} & \cup & \\[.2ex]
 \{ ( d_j, u_j - \rho, - ) \mid 1 \leq j \leq m \} & \cup & \{ ( c_i, \ell_i-1-\rho, + ) \mid 1 \leq i \leq n \} & \cup &  \{ ( 1, \rho, + ) \} \\[.2ex]
\end{array}\]
and where $0 \leq \rho < (\theta \cdot c)$.
% TIM: I do not think c is bound to anything at this point. I assume that c is the max of the c_i's
% and d_j's. I am not sure saying it is the selected c is good enough.
Since $c$ and $\theta$ are finite, there are a finite number of tuples of this form.
Since there are only a finite number of recursive calls to $\procsliah$ within $\procslia$,
and each call modifies $\vec t$ based a finite number of possible tuples coming from the set above,
the set of possible return values of $\procslia$ is finite, and thus it is finite for $\varphi[ \vec{ \con k }, \vec{\con e} ]$.
\qed
\end{proof}
\end{longv}

\begin{longv}
\begin{lem}\label{lem:procsliah-subs}
If $\I$ is a model for $\lia$ and a quantifier-free formula $\psi$, $\theta \geq 1$,
and $\procsliahh( \I, \psi, \con e, \theta ) = ( c, t, p )$, then:
\begin{enumerate}
\item \label{it:procsliah-subs-congruent}
$( c \cdot \con e )^\I \equiv_{\theta \cdot c} t^\I$, and
\item \label{it:procsliah-subs-model-preserve}
$\I \models \psi \{ c \cdot \con e \mapsto t \}$.
\end{enumerate}
\end{lem}
\begin{proof}
% TIM: I am retypesetting this to emphasize readability. Not sure I like it. Feel free to change it back.
We first show part~\ref{it:procsliah-subs-congruent}.
% in the case that $n>0$ and $\procsliahh( \I, \psi, \con e ) = ( c_i, \ell_i + \rho )$, we have
% $( \ell_i + \rho )^\I \ \opmod \ ( \theta \cdot c_i ) =$ 
% $( \ell_i + ( c_i \cdot \con e - \ell_i )^\I \ \opmod \ ( \theta \cdot c_i ) )^\I \ \opmod \ ( \theta \cdot c_i ) =$
% $( c_i \cdot \con e )^\I \ \opmod \ ( \theta \cdot c_i )$.
In the case that $n>0$ and $\procsliahh( \I, \psi, \con e, \theta ) = ( c_i, \ell_i + \rho, + )$, and
\[ ( \ell_i + \rho )^\I \equiv_{\theta \cdot c_i}
   ( \ell_i + ( c_i \cdot \con e - \ell_i )^\I \ \opmod \ ( \theta \cdot c_i ) )^\I \equiv_{\theta \cdot c_i}
   ( c_i \cdot \con e )^\I.\]
% In the case that $m>0$ and $\procsliahh( \I, \psi, \con e ) = ( d_j, u_j - \rho )$,  we have
% $( u_j - \rho )^\I \ \opmod \ ( \theta \cdot d_j ) =$ 
% $( u_j - (( u_j - d_j \cdot \con e )^\I \ \opmod \ ( \theta \cdot d_j )) )^\I \ \opmod \ ( \theta \cdot d_j ) =$
% $( d_j \cdot \con e )^\I \ \opmod \ ( \theta \cdot d_j )$.
In the case that $m>0$ and $\procsliahh( \I, \psi, \con e, \theta ) = ( d_j, u_j - \rho, - )$,  we have
\[
( u_j - \rho )^\I \equiv_{\theta \cdot d_j}
( u_j - ( u_j - d_j \cdot \con e )^\I \ \opmod \ ( \theta \cdot d_j ) )^\I \equiv_{\theta \cdot d_j}
( d_j \cdot \con e )^\I.\]
% In the case that $n=0$, $m=0$, and $\procsliahh( \I, \psi, \con e ) = ( 1, \rho )$, we have that
% $\rho^\I \ \opmod \ ( \theta \cdot 1 ) =$ 
% $( \con e^\I \ \opmod \ \theta )^\I \ \opmod \ ( \theta \cdot 1 ) =$
% $( 1 \cdot \con e )^\I \ \opmod \ ( \theta \cdot 1 )$.
In the case that $n=0$, $m=0$, and $\procsliahh( \I, \psi, \con e, \theta ) = ( 1, \rho, + )$, we have that
$\rho^\I \equiv_{\theta \cdot 1}$
$( \con e^\I \ \opmod \ \theta )^\I \equiv_{\theta \cdot 1}$
$( 1 \cdot \con e )^\I$.

To show part~\ref{it:procsliah-subs-model-preserve},
we first focus on the case where $n>0$ and $\procsliahh( \I, \psi, \con e, \theta ) = ( c_i, \ell_i + \rho, + )$.
We have that $\rho = ( c_i \cdot \con e - \ell_i )^\I \ \opmod \ ( \theta \cdot c_i )$.
Let $M$ be a set of literals of the form described in the body of $\procsliahh( \I, \psi, \con e )$.
We show that $\I$ satisfies each literal in $M \sigma$, where $\sigma = \{ c_i \cdot \con e \mapsto \ell_i + \rho \}$.
First, consider an atom in $M_\ell \sigma$ that is equivalent to
$( c_j \cdot \con e \geq \ell_j ) \sigma$ for some $j \in \{ 1, \ldots, n \}$.
This is equivalent to $( c_j \cdot c_i \cdot \con e \geq c_i \cdot \ell_j ) \sigma$, which is equivalent to
$\frac{c_j \cdot c_i}{c_i} \cdot ( \ell_i + \rho ) \geq \frac{c_j \cdot c_i}{c_j} \cdot \ell_j$,
which is satisfied by $\I$ since $( \frac{\ell_i}{ c_i } )^\I \geq ( \frac{\ell_j}{ c_j } )^\I$ by our selection of $( c_i, \ell_i + \rho)$ and since $\rho \geq 0$.
Second, consider the atom in $M_u \sigma$ that is equivalent to
$( d_j \cdot \con e \leq u_j ) \sigma$ for some $j \in \{ 1, \ldots, m \}$.
Let $\rho' = ( c_i \cdot \con e - \ell_i )^\I$, which is greater than $0$ since $\I$ satisfies $(c_i \cdot \con e \geq \ell_i)$.
Since $( c_i \cdot \con e )^\I = ( \ell_i + \rho' )^\I$,
we have that $\I$ satisfies $( d_j \cdot \con e \leq u_j ) \{ c_i \cdot \con e \mapsto \ell_i + \rho' \}$,
which is equivalent to $( d_j \cdot ( \ell_i + \rho' ) \leq c_i \cdot u_j )$.
Since $\rho = \rho' \ \opmod \ ( \theta \cdot c ) \leq \rho'$,
we have that $\I$ also satisfies 
$( d_j \cdot ( \ell_i + \rho ) \leq c_i \cdot u_j )$,
which is $( d_j \cdot \con e \leq u_j ) \sigma$.
Finally,  $\I$ satisfies $M_c \sigma$ as $M_c \sigma = M_c$ and $\I \models M_c$.
Thus, $\I$ satisfies $M \sigma$, which entails $\psi \sigma$.
The case for when $m>0$ and $\procsliahh( \I, \psi, \con e, \theta ) = ( d_j, u_j - \rho, - )$ is symmetric.
When $n=0$, $m=0$, and $\procsliahh( \I, \psi, \con e ) = ( 1, \rho, + )$, the assignment $M$ does not contain $\con e$, and thus $\I$ satisfies $M \{ c \cdot \con e \mapsto \rho \} = M$ and $\psi \{ c \cdot \con e \mapsto \rho \}$.
% TIM: I think this can be left implicit.
%Thus, the lemma holds.
\qed
\end{proof}
\begin{lem}\label{lem:procsliah}
Each recursive call to $\procsliah( \I, \psi, ( \con e_i, \ldots, \con e_n ), \theta, ( t_1, \ldots, t_{i-1} ), \vec p )$ 
within $\procslia( \I, \Gamma, ( \con e_1, \ldots, \con e_n ) )$
is such that:
\begin{enumerate}
\item \label{it:procsliah-inv-congruent}
$\theta \opdivides t_j^\I$ for each $1 \leq j < i$, and
\item \label{it:procsliah-inv-model-preserve}
$\I \models \psi$ and $\psi$ is equivalent to $\lnot \varphi[\vec{ \con k }, \vec{\con e}] \{ \theta \cdot \con e_1 \mapsto t_1 \} \cdot \ldots \cdot \{ \theta \cdot \con e_{i-1} \mapsto t_{i-1} \}$.
%$\I \models \Gamma \{ \con e_1 \mapsto t_1 \} \cdot \ldots \cdot \{ \con e_{i-1} \mapsto t_{i-1} \}$.
\end{enumerate}
\end{lem}
\begin{proof}
Both statements clearly hold for the initial call to $\procsliah$ in the body of $\procslia$.
Now, assume both statements hold for some call to $\procsliah( \I, \psi, ( \con e_i, \vec{\con e}' ), \theta, ( t_1, \ldots, t_{i-1} ), \vec p )$,
and assume $( c, t_i, p_i) = \procsliahh( \I, \psi, \con e_i, \theta )$.
We show that both statements hold for the call to 
$\procsliah( \I, \psi \sigma, \vec{\con e}', \theta \cdot c, ( ( c \cdot t_1 )\sigma, \ldots, ( c \cdot t_{i-1} )\sigma, \theta \cdot t_i ), (\vec p, p_i) )$,
where $\sigma = \{ c \cdot \con e_i \mapsto t_i \}$.

To show part~\ref{it:procsliah-inv-congruent},
we have from Lemma~\ref{lem:procsliah-subs} part \ref{it:procsliah-subs-congruent} that:
\begin{equation} \label{eqn:procsliah-subs-congruent}
( c \cdot \con e_i )^\I \equiv_{\theta \cdot c} ( t_i + \rho )^\I
\end{equation}
Consider a $t_j$ where $1 \leq j < i$, where 
by our assumption is such that $\theta \opdivides t_j^\I$,
and thus $\theta \cdot c \opdivides ( c \cdot t_j )^\I$.
% TIM: I do not think this step will be clear to the audience.
By (\ref{eqn:procsliah-subs-congruent}), we have that
$\theta \cdot c \opdivides ( ( c \cdot t_j )\sigma )^\I$.
Also by (\ref{eqn:procsliah-subs-congruent}), we have that $c \opdivides ( t_i + \rho )^\I$,
and thus $\theta \cdot c \opdivides ( \theta \cdot ( t_i + \rho ) )^\I$.

To show part~\ref{it:procsliah-inv-model-preserve},
by our assumption,
$\I \models \psi$ and thus by Lemma~\ref{lem:procsliah-subs} part \ref{it:procsliah-subs-model-preserve} we have that $\I \models \psi \sigma$.
By our assumption, $\psi$ is equivalent to 
$\lnot \varphi[\vec{ \con k }, \vec{\con e}] \{ \theta \cdot \con e_1 \mapsto t_1 \} \cdot \ldots \cdot \{ \theta \cdot \con e_{i-1} \mapsto t_{i-1} \}$.
%which is equivalent to $\psi \{ \theta \cdot c \cdot \con e_1 \mapsto c \cdot t_1 \} \cdot \ldots \cdot \{ \theta \cdot c \cdot \con e_{i-1} \mapsto c \cdot t_{i-1} \}$.
Thus,
$\psi \sigma$ is equivalent to
$\lnot \varphi[\vec{ \con k }, \vec{\con e}] \{ (\theta \cdot c) \cdot \con e_1 \mapsto ( c \cdot t_1 ) \sigma \} \cdot \ldots \cdot \{ (\theta \cdot c) \cdot \con e_{i-1} \mapsto ( c \cdot t_{i-1} ) \sigma \} \cdot \{ (\theta \cdot c) \cdot \con e_i \mapsto \theta \cdot ( t_i + \rho ) \}$.
Thus, the lemma holds.
\qed
\end{proof}
\end{longv}

\begin{lem}\label{lem:lia-monotonic}
$\procslia$ is model-preserving for $\varphi[ \vec{ \con k }, \vec{ \con e} ]$.
\end{lem}
\begin{longv}
\begin{proof}
Assume that $\procslia( \I, \Gamma, \lnot \varphi[\vec{ \con k }, \vec{\con e}] ) = \vec{t}$,
where $\vec{\con e} = ( \con e_1, \ldots, \con e_n )$, and $\vec{t} = ( t_1, \ldots, t_n )$.
By Lemma~\ref{lem:procsliah} and the definition of $\procslia$, there is a $\theta$ such that for each $i = 1, \ldots, n$,
term $t_i$ is of the form $s_i\ \opdivd^p \ \theta$ where $\theta \opdivides s_i^\I$,
and $\I \models ( \lnot \varphi[\vec{ \con k }, \vec{\con e}] ) \{ \theta \cdot \con e_1 \mapsto s_1 \} \cdot \ldots \cdot \{ \theta \cdot \con e_n \mapsto s_n \}$.
Thus, $\I$ satisfies $( \lnot \varphi[\vec{ \con k }, \vec{\con e}] ) \{ \vec{ \con e } \mapsto \vec{t} \} =$ $\lnot \varphi[\vec{ \con k }, \vec{t}]$,
and thus $\procslia$ is model-preserving for $\varphi[ \vec{ \con k }, \vec{ \con e} ]$.
%Assume by contradiction $\procslia( \I, \Gamma, \lnot \varphi[\vec{ \con k }, \vec{\con e}] ) = \vec{t}$, 
%where $\varphi[\vec{\con k}, \vec{t}] \in \Gamma$.
%By definition of selection function, we have $\I$ is a model of $\lia$, $\Gamma$ and $\lnot \varphi[\vec{ \con k }, \vec{\con e}]$.
%Let $\vec{\con e} = ( \con e_1, \ldots, \con e_n )$ and $\vec{t} = ( t_1, \ldots, t_n )$.
%By Lemma~\ref{lem:procsliah}, there is a $\theta$ such that for each $i = 1, \ldots, n$,
%term $t_i$ is of the form $\frac{ s_i }{ \theta }$ where $\theta \opdivides s_i^\I$,
%and $\I \models ( \lnot \varphi[\vec{ \con k }, \vec{\con e}] ) \{ \theta \cdot \con e_1 \mapsto s_1 \} \cdot \ldots \cdot \{ \theta \cdot \con e_n \mapsto s_n \}$.
%Thus, $\I$ satisfies $( \lnot \varphi[\vec{ \con k }, \vec{\con e}] ) \{ \vec{ \con e } \mapsto \vec{t} \} =$ $\lnot \varphi[\vec{ \con k }, \vec{t}]$.
%This is a condradicition since 
%$\varphi[\vec{\con k}, \vec{t}] \in \Gamma$.
%Thus, $\procslia( \I, \Gamma, \vec{\con e} ) \neq \vec{t}$ and
%$\procslia$ is monotonic for $\varphi[ \vec{ \con k }, \vec{ \con e} ]$.
\qed
\end{proof}
\end{longv}

\begin{thm}
$\mathcal{P}_{\procslia}$ is a sound and complete procedure 
for determining the $\lia$-satisfiability of $\exists \vec a\, \forall \vec x\, \varphi[\vec a, \vec x]$.
\end{thm}
\begin{proof}
By Theorem~\ref{lem:sel-fin-mono}, Lemma~\ref{lem:sel-mm-m}, Lemma~\ref{lem:lia-finite} and Lemma~\ref{lem:lia-monotonic}.
\qed
\end{proof}

\ 

%We now demonstrate the procedure with a few examples.

% TIM: The ceiling function in t[k] and the terms added to \Gamma in the examples are definitely going to confuse readers. We need to introduce this more clearly before using it.
% TIM: Also the description the models need to come sooner. It is confusing to the execution in example 9 without it.
% TIM: Note that in theses longer examples, the left hand columns are not that informative. it just tells you the iteration number (which is obvious) and "sat  sat".

\begin{example}
To demonstrate a case involving a substitution with coefficients,
consider the formula $\forall xy\, ( 2 \cdot x < a \lor x + 3 \cdot y < b )$
whose negation is $2 \cdot e_1 \geq a \land e_1 + 3 \cdot e_2 \geq b$.
A possible run of $\mathcal{P}_{\procslia}$ on this input is as follows.

\noindent
\resizebox{\columnwidth}{!}{$
\begin{array}{c@{\hspace{1em}}c@{\hspace{1em}}c@{\hspace{1em}}|c@{\hspace{1em}}c@{\hspace{1em}}c@{\hspace{1em}}c@{\hspace{1em}}|c@{\hspace{1em}}l}
\hline
 &  &  & \multicolumn{4}{c|}{\procsliahh( \I, \Gamma, \con e, \theta )} & &  \\
\text{\#} & \Gamma & \Gamma' & \con{e} & \theta & M_\ell  & \text{return} & \vec t[\vec{\con k}] & \text{Add to } \Gamma  \\
\hline
 1 & \text{sat} & \text{sat} & 
 e_1 & 1 & \{ \underline{2 \cdot e_1 \geq a}, \ldots \} & (2, a, +) & 
 \\
  & & & 
  e_2 & 2 &\{ \underline{6 \cdot e_2 \geq 2 \cdot b - a} \} & (6, 2 \cdot b - a, +) & ( 6 \cdot a, 4 \cdot b - 2 \cdot a )\ \opdivd^{+} \ 12& \psi_1
 \\
 2 & \text{unsat} & & & & & &
 \\
\hline
\end{array}$
}

\ 

\noindent
Thus, $\exists ab\, \forall xy\, ( 2 \cdot x < a \vee x + 3 \cdot y < b )$ is $\lia$-unsatisfiable.
We assume $\rho = 0$ for all calls to $\procsliahh$ in this run.
Applying the substitution $\{ 2 \cdot e_1 \mapsto a \}$
to $e_1 + 3 \cdot e_2 \geq b$ results in the bound $6 \cdot e_2 \geq 2 \cdot b - a$ for $e_2$.
We add to $\Gamma$ the instance $\psi_1$,
which is equivalent to $2 \cdot ((6 \cdot a)\ \opdivd^{+} \ 12 ) < a \lor (6 \cdot a)\ \opdivd^{+} \ 12 + 3 \cdot ((4 \cdot b - 2 \cdot a)\ \opdivd^{+} \ 12) < b$.
Applying normalization $\normalize{}$ to this formula results in a one that is $\lia$-unsatisfiable.
%Assuming $\normalize{}$ interprets integer division as rounding up,
%we have $\psi_1$ is unsatisfiable.
\qed
\end{example}

\begin{example}
To demonstrate a case involving a non-zero value of $\rho$,
consider the formula $\forall xy\, ( 3 \cdot x + y \not\teq a \vee 0 > y \vee y > 2 )$
whose negation is $3 \cdot e_1 + e_2 \teq a \wedge 0 \leq e_2 \wedge e_2 \leq 2$,
where $\teq$ denotes the conjunction of non-strict upper and lower bounds.
A possible run of $\mathcal{P}_{\procslia}$ on this input is as follows.

\noindent
\resizebox{\columnwidth}{!}{$
\begin{array}{c@{\hspace{1em}}c@{\hspace{1em}}c@{\hspace{1em}}|c@{\hspace{1em}}c@{\hspace{1em}}c@{\hspace{1em}}c@{\hspace{1em}}|c@{\hspace{1em}}l}
\hline
 &  &  & \multicolumn{4}{c|}{\procsliahh( \I, \Gamma, \con e, \theta )} & &  \\
\text{\#} & \Gamma & \Gamma' & \con{e} & \theta & M_\ell  & \text{return} & \vec t[\vec{\con k}] & \text{Add to } \Gamma  \\
\hline
 1 & \text{sat} & \text{sat} & 
 e_1 & 1 & \{ \underline{3 \cdot e_1 \geq a-e_2} \} & (3, a-e_2, +) & 
 \\
  & & & 
 e_2 & 3 & \{ \underline{e_2 \geq 0} \} & (1, 0, +) & ( a, 0 )\ \opdivu \ 3 & \psi_1
 \\
  2 & \text{sat} & \text{sat} & 
 e_1 & 1 & \{ \underline{3 \cdot e_1 \geq a-e_2} \} & (3, a-e_2, +) & 
 \\
  & & & 
 e_2 & 3 & \{ \underline{e_2 \geq 0} \} & (1, 1, +) & ( a-1, 1 )\ \opdivu \ 3 & \psi_2
 \\
  3 & \text{sat} & \text{sat} & 
 e_1 & 1 & \{ \underline{3 \cdot e_1 \geq a-e_2} \} & (3, a-e_2, +) & 
 \\
  & & & 
 e_2 & 3 & \{ \underline{e_2 \geq 0} \} & (1,2, +) & ( a-2, 2 )\ \opdivu \ 3 & \psi_3
 \\
 4 & \text{unsat} & & & & & &
 \\
\hline
\end{array}$
}

\ 

\noindent
This run shows $\exists a\, \forall xy\, ( 2 \cdot x < a \vee x + 3 \cdot y < b )$ is $\lia$-unsatisfiable.
On the first iteration, we assume that $\Gamma'$ is satisfied by a model, call it $\I_1$, that interprets all variables as $0$,
and hence the values chosen for $e_1$ and $e_2$ correspond to their maximal lower bounds in $\I_1$, $a-e_2$ and $0$ respectively,
where in each call to $\procsliahh$ we have $\rho = 0$.
The instance $\psi_1$ added to $\Gamma$ on this iteration is equivalent to
$3 \cdot (a\ \opdivu \ 3) \not\teq a$ and
implies that $a^{\I} \not\equiv_3 0$ in subsequent models $\I$.
Thus, models $\I$ satisfying $3 \cdot e_1 + e_2 \teq a$
are such that $e_2^{\I} \not\equiv_3 0$.
On the next iteration, $\Gamma'$ is satisfied by a model, call it $\I_2$,
where the maximal lower bound for $e_2$ is $0$. 
By the above reasoning and since $\I_2$ satisfies $3 \cdot e_1 + e_2 \teq a$, 
it must be that $\rho = ((e_2-0)^{\I_2} \ \opmod \ 3) \neq 0$.
Assume $(e_2-0)^{\I_2} \equiv_3 1$.
The instance $\psi_2$ is equivalent to
$3 \cdot ((a-1)\ \opdivu \ 3)  + 1 \not\teq a$,
which implies that $a^{\I} \not\equiv_3 1$ in subsequent models $\I$,
and hence $e_2^{\I} \not\equiv_3 1$.
The instance $\psi_3$ is equivalent to
$3 \cdot ((a-2)\ \opdivu \ 3) + 2 \not\teq a$ and
implies that $a^{\I} \not\equiv_3 2$, which together with the 
two previous instances are $T$-unsatisfiable.
\qed
\end{example}

The procedure $\mathcal{P}_{\procslia}$ can be understood to lazily 
enumerating disjuncts of Cooper's algorithm for quantifier elimination over linear integer arithmetic~\cite{cooper1972},
with minor differences.
\begin{longv}
The algorithm is essentially enumerating a single path of \cite{cooper1972} by using the model to select a satisfied case split for each variable over an entire block of quantifiers.\footnote{
In the parlance of~\cite{cooper1972}, $\mathcal{P}_{\procslia}$ selects a feasible $j$ value using the calculation of $\rho$ 
and avoids introducing the $F_{\pm\infty}$ cases by introducing the no bounds case ($n = 0, m=0$) and always favoring bounds when one exists.
  }
\end{longv}
Like that approach,
the worst-case performance is dependent upon the size of coefficients of monomials,
which is manifested in our case by the fact that the number of possible return values of $\procsliahh$ is proportional to the size of $\theta$.
While not shown here, our implementation takes steps to reduce the size of $\theta$
by factoring out common divisors in $\theta$ and the coefficients returned by $\procsliahh$.

\begin{longv}
\subsection{Comparison to Existing Approaches}
The approach taken in $\mathcal{P}_{\procslia}$ is similar to the one taken in
Section~2.5 in \cite{bjornerplaying}. The most substantial difference between the two algorithms is that
$\mathcal{P}_{\procslia}$ implements a variant of Cooper's algorithm while $\mathit{resolve}$ in~\cite{bjornerplaying} uses the model to guide an execution of the Omega test~\cite{Pugh91OmegaTest}.
The most similar aspects of the approaches are the computation of a feasible $\rho$ and the computation of the $d$ values in the \emph{grey} shadow cases of $\mathit{resolve}$.
These differ in that a different $d$ value is selected to ensure separation between each upper bound and the greatest lower bound in a \emph{projection} whereas $\rho$ is selected using the current value of $c_i \cdot e$ (the selection of $d$ is agnostic to $e$ in our parlance) to ensure all bounds are satisfied by a single instantiation.

\end{longv}

%\section{Handling Arbitrary $T$-formulas}
\section{Integration in an SMT Solver}
\label{sec:smt}

This section gives an overview of
how the instantiation-based procedure for quantified formulas as described in Section~\ref{sec:qi}
can be integrated into a solving architecture used by SMT solvers, and
used in part for determining the satisfiability of inputs $\Gamma_0$
having arbitrary Boolean structure that contain any number of quantified $T$-formulas.
%We describe the high-level details of the approach only,
%and comment informally on its properties here.

\begin{figure}[t]
\begin{framed}
% TIM: I would phrase this as a Repeat and a loop just to make it look like Figure 1 as much as possible.
$\smtsolve( T, \Gamma )$:
\begin{itemize}
\item[\ ] If $\Gamma$ is $T$-unsatisfiable, then return ``unsat".
\item[\ ] Otherwise, 
\begin{itemize}
\item[\ ] Let $M$ be such that $\I \models M$ and $M \models_p \Gamma$ for some model $\I$ of $T$ and $\Gamma$.
\item[\ ] $\Gamma' := \Gamma$.
\item[\ ] For each $( \neg ) A_i \in M$ where $A_i \Leftrightarrow \forall \vec{x}. \varphi_i[ \vec{x} ]$,
\begin{itemize}
\item[\ ] $\Gamma' := \Gamma' \cup ( A_i \vee B_i ) \cup ( B_i \Rightarrow \lnot \varphi_i[ \vec{e}_i ] )$.
% TIM: Why would A_i/B_i be in M already? M has not changed in the larger repeat loop.
% If there is a delay until the next round to instantiate, this needs to be explained.
\item[\ ] If $A_i \in M$ and $B_i \in M$,
\begin{itemize}
% TIM: Why not always set A_i <- true and B_i <- false? If one did that, would this case ever get riggered?
\item[\ ] $\Gamma' := \Gamma' \cup ( A_i \Rightarrow \varphi_i[ \mathcal{S}_i( \I, \Gamma, \lnot \varphi_i[ \vec{e}_i ] ) ] )$
\end{itemize}
%\item[\ ] If $\neg A_i \in M$,
%\begin{itemize}
%\item[\ ] $\Gamma' := \Gamma' \cup ( \neg A_i \Rightarrow B_i )$
%\end{itemize}
\end{itemize}
\item[\ ] If $\Gamma' = \Gamma$, then return ``sat". Otherwise, return $\smtsolve( T, \Gamma' )$.
\end{itemize}
\end{itemize}
\end{framed}
\vspace{-2ex}
\caption{Procedure $\smtsolve$ for SMT solving with quantifier instantiation, 
which determines the $T$-satisfiability of a ground set of $T$-formulas $\Gamma$ in purified form with respect to quantified formulas.
\label{fig:proc-smt-qi-solve}}
\end{figure}

Figure~\ref{fig:proc-smt-qi-solve} defines a procedure $\smtsolve$ which takes as input a theory $T$ 
and a set of ground $T$-formulas $\Gamma$ in \emph{purified form} with respect to quantified formulas, that is,
% TIM: "In other words" is inappropraite. The reader doesn't know what purified means until this sentence.
a formula obtained from $\Gamma_0$ 
by replacing all quantified formulas $\forall \vec{x}. \varphi_i[ \vec{x} ]$ in $\Gamma_0$ by (uniquely associated) boolean variables $A_i$.
We call such a variable $A_i$ the \emph{positive guard} of $\forall \vec{x}. \varphi_i[ \vec{x} ]$
and write $A_i \Leftrightarrow \forall \vec{x}. \varphi_i[ \vec{x} ]$
to denote $A_i$ is the positive guard of $\forall \vec{x}. \varphi_i[ \vec{x} ]$.
In each call to $\smtsolve$,
if $\Gamma$ is $T$-unsatisfiable, then the procedure returns ``unsat".
Otherwise, we find a model $\I$ of $T$ and $\Gamma$ and a corresponding propositionally satisfying assignment $M$.
We then build a new set of formulas $\Gamma'$, initially containing $\Gamma$, as follows.
For each quantified formula $\forall \vec{x}. \varphi_i[ \vec{x} ]$ whose positive guard $A_i$ is in $M$,
we add the formulas $( A_i \vee B_i )$ and $( B_i \Rightarrow \lnot \varphi_i[ \vec{e}_i ] )$ to $\Gamma'$ if we have not done so already,
where $B_i$ is a fresh Boolean variable, which we call the \emph{negative guard} of $\forall \vec{x}. \varphi_i[ \vec{x} ]$.
If both the positive and negative guards of $\forall \vec{x}. \varphi_i[ \vec{x} ]$ are asserted positively in $M$ 
(we will say such a quantified formula is \emph{active in $M$}),
we consider an instance of this quantified formula to $\Gamma'$ based on its associated selection function $\mathcal{S}_i$,
If $\varphi_i$ has nested quantification, this instance will contain quantified formulas, 
which we purify in the same manner described above.
If no new formulas are added to $\Gamma'$ in this process, then the procedure returns ``sat".
Otherwise, we call $\smtsolve$ on $\Gamma'$.

At a high level,
the procedure in Figure~\ref{fig:proc-smt-qi-solve} 
adds guarded instances of quantified formulas to an evolving set of formulas $\Gamma$ 
until $\Gamma$ is $T$-unsatisfiable, or a fixed point is reached.
In this respect,
the algorithm is similar to existing instantiation-based approaches used by SMT solvers for 
quantified formulas~\cite{MouraBjoerner07EfficientEmatchingSmtSolvers,GeBarrettTinelli07SolvingQuantifiedVerificationConditionsUsingSatisfiability}.
However, the procedure differs from these approaches in the following ways.
Firstly, when a universally quantified formula are asserted negatively in $M$,
typical approaches add the clause $( \neg A_i \Rightarrow \lnot \varphi_i[ \vec{e}_i ] )$ to $\Gamma$.
Here, we choose to add the clauses
$( A_i \vee B_i )$ and $( B_i \Rightarrow \lnot \varphi_i[ \vec{e}_i ] )$ instead.
This allow us to consider both the positive and negative versions of quantified formulas simultaneously.
As such, the algorithm only adds instances of quantified formulas 
where \emph{both} the positive and negative guards are asserted positively in $M$.
\begin{longv}
To ensure the model soundness of the approach (that is, the algorithm answers ``sat" only if the input is indeed $T$-satisfiable),
we require the following property of the set of literals $M$ in the body of $\smtsolve$:
% TIM: I think we may need to prove this.
% AJR: agreed, but let's not worry about it for the paper version.
\begin{align}
\text{If no quantified formula is active in $M$, then $\Gamma \cup \{ B_i \}$ is $T$-unsat for $i = 1, \ldots, n$,}\nonumber \\
\text{where $\{ A_1, \ldots, A_n \}$ is the set of the positive guards that are asserted positively in $M$.}\nonumber
\end{align}
In other words, 
the set of literals $M$ are chosen such that,
if possible, at least one of $B_1, \ldots, B_n$ is true in $M$.
In practice, this requirement can be met in a DPLL(T)-based SMT solver
by instructing its underlying SAT solver, when it chooses a decision literal to add to $M$, to choose
one of the unassigned negative guards of quantified formulas in $\Gamma$, if one exists, and assert it positively.
% TIM: I suspect the previous two sentences will not make sense to non-CVC4 developers.
% This discussion probably needs to get pushed into the algorithm.

We refer to the treatment of quantified formulas in Figure~\ref{fig:proc-smt-qi-solve} as \emph{counterexample-guided quantifier instantiation}~\cite{ReynoldsDKBT15Cav}.
A closely related approach is that of model-based quantifier instantiation~\cite{GeDeM-CAV-09},
which like the approach described here,
adds instances of quantified formulas based on models for their negations.
This approach differs in its scope,
in that it primarily targets quantified formulas having uninterpreted functions,
whereas the approach described in Figure~\ref{fig:proc-smt-qi-solve} targets quantified formulas having no uninterpreted functions.
It also differs in that it uses a separate copy of the SMT solver as an oracle for checking the satisfiability of the negation of each quantified formula it instantiates,
whereas the approach described in Figure~\ref{fig:proc-smt-qi-solve} uses a single instance of the SMT solver for doing these tasks simultaneously in its main solving loop.

\subsection{Arbitrary Quantifier Alternation}
\label{sec:nested-quant}
% TIM: Does this need its own section?

The algorithm in Figure~\ref{fig:proc-smt-qi-solve}
can be used as a basis for handling quantified formulas with arbitrary quantifier alternations.
Assume we rewrite formulas added to $\Gamma'$ in the body $\smtsolve$ so they are in purified form with respect to quantified formulas,
that is, we replace each quantified formula $\forall \vec{x}. \varphi_j[ \vec{x} ]$ occurring in formulas 
$B_i \Rightarrow \lnot \varphi_i[ \vec{e}_i ]$ and $A_i \Rightarrow \varphi_i[ \mathcal{S}_i( \I, \Gamma, \vec{e}_i ) ]$ with its corresponding positive guard $A_j$.
We give an intuition of how the procedure $\smtsolve$ handles such formulas in the following.
A more comprehensive description is the subject of future work.

Consider the $\lia$-formula $\forall x. \neg ( \forall y. x > y )$, call it $\varphi_1$, whose positive guard is $A_1$.
Given the input $\Gamma = \{ A_1 \}$, the procedure $\smtsolve$
adds the formulas $A_1 \vee B_1$ and $B_1 \Rightarrow A_2$ to $\Gamma$,
where $A_2$ is the positive guard for $(\forall y. e_1 > y)$ (call this formula $\varphi_2$) where $e_1$ is a fresh constant.
On the second call to $\smtsolve$, 
the satisfying assignment includes $A_2$ and
% TIM: The assignment must also contain $B_2$?
the procedure similarly
adds the formulas $A_2 \vee B_2$ and $B_2 \Rightarrow e_1 \leq e_2$ to $\Gamma$ where $e_2$ is a fresh constant.
On the third call to $\smtsolve$, 
we have that $\Gamma = \{ A_1, B_1 \Rightarrow A_2, B_2 \Rightarrow e_1 \leq e_2, \ldots \}$ is $T$-satisfiable,
and we may choose a satisfying assignment $M = \{ A_1, B_1, A_2, B_2, e_1 \leq e_2, \ldots \}$.
Both $\varphi_1$ and $\varphi_2$ are active in $M$.
The literal $e_1 \leq e_2$ is over the atoms of $\varphi_2[e_2/y]$, 
and we may add the formula $A_2 \Rightarrow e_1 > e_1$ to $\Gamma$ on this iteration,
assuming our selection function for $\varphi_2$ chose to return the maximal lower bound $e_1$ for $e_2$.
On the fourth call to $\smtsolve$, 
we have that $\Gamma = \{ A_1, B_1 \Rightarrow A_2, A_2 \Rightarrow e_1 > e_1, \ldots \}$, and
hence $B_1$ or $A_2$ cannot be asserted positively in a satisfying assignment $M$ for this set.
Hence, neither $\varphi_1$ nor $\varphi_2$ is active in $M$ and
the procedure $\smtsolve$ adds no instances to $\Gamma$, indicating that our input is satisfiable.
\end{longv}
\begin{shortv}
Although not shown here, we require that the $M$ in Figure~\ref{fig:proc-smt-qi-solve} is chosen such that,
whenever possible, at least one quantified formula is active in $M$.%~\cite{RKK2015}.

We refer to the treatment of quantified formulas in Figure~\ref{fig:proc-smt-qi-solve} as \emph{counterexample-guided quantifier instantiation}~\cite{ReynoldsDKBT15Cav}.
A related approach is that of model-based quantifier instantiation~\cite{GeDeM-CAV-09},
which like the approach described here,
adds instances of quantified formulas based on models for their negations.
This approach differs in its scope,
in that it targets quantified formulas having uninterpreted functions,
whereas the approach described in Figure~\ref{fig:proc-smt-qi-solve} targets quantified formulas having no uninterpreted functions.
It also differs in that it uses a separate copy of the SMT solver as an oracle for checking the satisfiability of the negation of each quantified formula it instantiates,
whereas the approach described in Figure~\ref{fig:proc-smt-qi-solve} uses a single instance of the SMT solver for doing these tasks simultaneously in its main solving loop.
\end{shortv}

\begin{comment}
We claim without proof the procedure $\smtsolve$ is sound for determining the (un)satisfiability of arbitrary $T$-formulas $\Gamma$.
% TIM: I think it is important to hammer that these are separate claims.
We additionally claim that when all quantified formulas in $\Gamma$ are such that 
their bodies reside in a language $\lan$ that is closed under negation and where the satisfiability of finite sets of $\lan$ formulas modulo $T$ is decidable, and 
have corresponding selection functions that are finite and monotonic, 
then $\smtsolve$ is terminating as well.
% TIM: I am not sure I buy this as stated. I think this requires that the bodies of the quantifiers of the formulas in gamma stay in \lan. The selection function for LIA for example introduces divs which slightly shifts the logic.
% TIM: Question: What happens if the formula is \forall x \exists y \phi. The expansion would be B => \lnot \exists y \phi. Is an A variable created for \forall y. \lnot \phi ? 
This includes inputs in both $\lia$ and $\lra$ having arbitrary Boolean structure, but where nested quantification is not allowed.
\end{comment}

\section{Experimental Evaluation}
\label{sec:results}

We have implemented the procedure in the SMT solver \cvc~\cite{CVC4-CAV-11} (version 1.5 pre-release).
This section presents an evaluation of this implementation compared against other SMT solvers
and first-order theorem provers.

\paragraph{Pure Quantified Linear Arithmetic}
We considered all quantified benchmarks over 6 classes in the $\lra$ and $\lia$ logics of the SMT library~\cite{BarST-SMTLIB}.
The class 
{\bf keymaera} are verification conditions coming from the Keymaera verification tool~\cite{platzer2009real},
{\bf scholl} were used for simplification of non-convex polyhedra in~\cite{scholl2008using},
{\bf psyco} were used for weakest precondition synthesis for compiler optimizations in~\cite{LopesM14},
{\bf uauto} correspond to verification conditions in~\cite{HeizmannTACAS2015},
and the {\bf tptp} classes correspond to simple arithmetic conjectures coming from the TPTP library~\cite{SS98}.
\begin{longv}
We also considered a class of benchmarks {\bf sygus} corresponding to 
first-order formulations of the 71 single-invocation synthesis conjectures taken from 
the conditional linear integer track of the 2015 edition of the syntax-guided synthesis competition~\cite{AlurETAL2014SyGuSMarktoberdorf}.
\end{longv}
\begin{shortv}
We also considered a class of benchmarks {\bf sygus} corresponding to first-order formulations of conjectures taken from 
the conditional linear integer track of the 2015 edition of the syntax-guided synthesis competition~\cite{AlurETAL2014SyGuSMarktoberdorf}.
\end{shortv}
% TIM: Do readers care about classes? If so, we may need to say why. I do not know for example.
\begin{longv}
All benchmarks are in the SMT version 2 format.
For comparisons with automated theorem provers, they were converted to the TPTP format by the SMTtoTPTP conversion tool~\cite{SMTtoTPTP2015}.
\end{longv}
We remark that all benchmarks
consist purely of quantified formulas over linear arithmetic with very little, and in a majority of cases, no quantifier-free content.
\footnote{
\begin{longv}
Details can be found at {\scriptsize \url{http://cs.uiowa.edu/~ajreynol/InstLA}}.
\end{longv}
\begin{shortv}
Details can be found at {\scriptsize \url{http://cs.uiowa.edu/~ajreynol/IJCAR2016-InstLA}}.
\end{shortv}
}

\begin{figure}[t]
\centering
{
\begin{tabular}{|l|cc|cc|cc|cc|}
\hline                                                                
  & \multicolumn{2}{c|}{{\bf keymaera} (222)}   & \multicolumn{2}{c|}{{\bf scholl} (371)}     
  & \multicolumn{2}{c|}{{\bf tptp} (25)}   & \multicolumn{2}{c|}{{\bf Total} (621)}    
\\                                                                
  & \#  & time  & \#  & time  & \#  & time  & \#  & time
\\                                                                
\hline                                                                
{\bf  CVC4(a)  } & \bf 222 & 2.0 & \bf 352 & 2176.4  & \bf 25  & 0.2 & \bf 599 & 2178.6  \\
{\bf  CVC4  } & \bf 222 & 2.0 & 351 & 1074.0  & \bf 25  & 0.2 & 598 & 1076.2  \\
{\bf  Z3  } & \bf 222 & 2.4 & 326 & 553.0 & \bf 25  & 0.4 & 573 & 555.8 \\
{\bf  VampireZ3 } & 220 & 51.2  & 57  & 393.2 & \bf 25  & 2.3 & 302 & 446.7 \\
{\bf  Beagle  } & \bf 222 & 377.9 & 53  & 577.9 & \bf 25  & 29.7  & 300 & 985.5 \\
{\bf  Vampire } & 218 & 57.8  & 43  & 31.1  & \bf 25  & 1.3 & 286 & 90.1  \\
{\bf  Yices } & \bf 222 & 0.9 & -- & 0.0 & \bf 25  & 0.01  & 247 & 1.0 \\
{\bf  ZenonArith  } & 205 & 13.8  & 25  & 452.9 & 14  & 0.9 & 244 & 467.7 \\
{\bf  Princess  } & 202 & 1136.2  & 0 & 0.0 & \bf 25  & 67.4  & 227 & 1203.6  \\                                     
\hline                                                                
\end{tabular}
\\
}
\caption{Results for $\lra$ benchmarks, showing times 
(in seconds) and benchmarks solved
by each solver and configuration over 3 benchmark classes with a 300s timeout.
Yices (version 2.4.1) does not support nested quantification,
hence it was not applicable for the scholl class.}
\label{fig:results-lra}
\end{figure}

\begin{figure}[t]
\centering
{
\begin{tabular}{|l|cc|cc|cc|cc|cc|}
\hline                                                                
  & \multicolumn{2}{c|}{{\bf psyco} (189)}   & \multicolumn{2}{c|}{{\bf tptp} (46)}     
  & \multicolumn{2}{c|}{{\bf uauto} (155)}      & \multicolumn{2}{c|}{{\bf sygus} (71)}    & \multicolumn{2}{c|}{{\bf Total} (461)}    
\\                                                                
  & \#  & time  & \#  & time  & \#  & time  & \#  & time & \#  & time 
\\                                                                
\hline                                                                
{\bf  CVC4  } & \bf 189 & 78.7  & \bf 46  & 0.4 & \bf 155 & 1.9 & \bf 71  & 22.0  & \bf 461 & 103.0 \\
{\bf  Z3  } & 183 & 32.1  & \bf 46  & 0.7 & \bf 155 & 1.8 & \bf 71  & 19.0  & 455 & 53.6  \\
{\bf  Beagle  } & 28  & 900.0 & \bf 46  & 48.4  & 153 & 343.6 & 57  & 617.7 & 284 & 1909.7  \\
{\bf  Princess  } & 13  & 513.4 & \bf 46  & 48.0  & \bf 155 & 201.9 & 68  & 418.8 & 282 & 1182.1  \\
{\bf  VampireZ3 } & 4 & 3.1 & 36  & 4.7 & \bf 155 & 106.3 & 55  & 151.8 & 250 & 265.9 \\
{\bf  Vampire } & 6 & 196.0 & 36  & 2.0 & \bf 155 & 378.0 & 46  & 262.8 & 243 & 838.7 \\
{\bf  ZenonArith  } & 0 & 0.0 & 30  & 1.9 & 154 & 15.0  & 28  & 1374.6  & 212 & 1391.5  \\                                                      
\hline                                                                
\end{tabular}
\\
}
\caption{Results for $\lia$ benchmarks, showing times 
(in seconds) and benchmarks solved
by each solver and configuration over 4 benchmark classes with a 300s timeout.}
\label{fig:results-lia}
\end{figure}

The results for the linear real and integer benchmarks are in Figures~\ref{fig:results-lra} and~\ref{fig:results-lia} respectively.
Of the 7 benchmark classes,
only one (the {\bf scholl} class from $\lra$) had quantified formulas with nested quantification.
The algorithm in Section~\ref{sec:smt} naturally extends to such formulas;
a formal treatment of nested quantification is the subject of future work.
%As alluded to in Section~\ref{sec:smt}, our 
%implementation is capable of handling such problems by applying quantifier instantiation
%to quantified formulas as they become introduced by clauses added to $\Gamma$ in Figure~\ref{fig:proc-smt-qi-solve}.

For $\lra$,
we considered both the selection function from Figure~\ref{fig:sel-lra}, and its alternative from Figure~\ref{fig:sel-lra2},
where the latter we refer to as \cvc(a).
For both $\lra$ and $\lia$, the best configuration of \cvc solves the most benchmarks overall (599 and 461 respectively),
and did not give a conflicting response with any of the other solvers.
Moreover, we note that \cvc solves \emph{every} benchmark that does not involve nested quantification,
giving confirming evidence that our approach and implementation for solving
linear arithmetic with one quantifier alternation is indeed sound and complete.
Although we do not claim completeness for formulas with nested quantification,
\cvc solves more benchmarks (352) from the {\bf scholl} class than any other solver.

The SMT solver \ziii (version 4.3.2), which uses the approach described in~\cite{Bjoerner10LinearQuantifierEliminationAsAbstractDecision},
solves the next most benchmarks overall, 
solving 573 and 455 total for the $\lra$ and $\lia$ sets respectively.
A technique~\cite{dutertresolving} in the SMT solver Yices (version 2.4.1)
is able to solve all benchmarks from the {\bf keymaera} and {\bf tptp} classes of $\lra$,
both does not handle quantified formulas in $\lia$ or with nested quantification. 
We also considered the entrants of the first-order typed theorems division (TFA) of CASC 25,
the most recent competition for automated theorem provers~\cite{sutcliffe20157th}.
\begin{longv}
\footnote{
We omit SPASS+T, which did not handle some classes of benchmarks due to restrictions on its input format,
was comparable to the other automated theorem provers for the others. 
We show results for an updated version of Beagle (version 0.9.30).}\end{longv}
For both benchmarks over $\lia$ and $\lra$, the automated theorem provers trail the performance of \cvc (and \ziii) significantly.
The best $\lra$ automated theorem prover, VampireZ3, which uses a combination of a first-order theorem prover and an SMT solver~\cite{reger2015playing}, 
solves only 302 benchmarks, compared to 599 solved by \cvc(a).
The best $\lia$ automated theorem prover, Beagle~\cite{baumgartner2015}, solves 284 benchmarks, also notably less than the 461 solved by \cvc.

When comparing the best configuration of \cvc to a combination of all other solvers, 
\cvc solved 31 benchmarks that no other system solved, 
while in only 10 cases did another system solve a benchmark that \cvc could not solve.
In addition to solving the most benchmarks, \cvc generally has small runtimes for the benchmarks it solves.
\begin{longv}
When compared to the second best solver \ziii, which takes 1129.8 seconds to solve 1028 benchmarks over all classes,
\cvc while solving 1060 benchmarks overall, solves its first 1028 benchmarks in a total of 159.0 seconds.
\end{longv}
Among the benchmarks solved by \cvc and other systems, 
in only 12 cases did any system solve a benchmark at least 5 seconds faster than \cvc,
while in 13 cases \cvc solved a benchmark at least 5 seconds faster than all other systems.

\paragraph{Combining Linear Arithmetic with Uninterpreted Functions}
A prototype of the instantiation-based procedure from this paper was used by \cvc
in both the CASC J7 and CASC 25 competitions~\cite{sutcliffe20157th}, which evaluated automated theorem provers on
TPTP benchmarks involving combinations of arithmetic and free function symbols.
\cvc won the theorems (TFA) division of CASC J7 and finished 2$^{nd}$ in CASC 25, behind VampireZ3.
Additionally, \cvc won the non-theorems (TFN) division of CASC 25.
While treatment of uninterpreted functions is beyond the scope of this work,
this shows the potential of the technique for use in general first-order automated theorem proving in the presence of background theories.

%AJR: it would really get the point across if we say this algorithm was the key to winning SMTCOMP 2015 LIA/LRA, SygusComp General/CLIA, CASC J7 TFA, CASC 25 TFN.

\section{Conclusion}
\label{sec:conclusion}

We have presented a class of instantiation-based procedures that 
are at the same time complete for quantified linear arithmetic and highly efficient in practice. 
Thanks to our framework we also obtain a simple
and modular correctness argument for soundness and completeness on formulas with one quantifier alternation.

For future work, 
we would like to adapt the approach for quantified linear arithmetic with
arbitrary quantifier alternations, and develop heuristics for avoiding
worst case performance for quantified integer arithmetic involving large coefficients.
We plan to develop selection functions for other theories,
in particular, algebraic datatypes and fixed-width bitvectors,
as well as for combinations of theories that admit quantifier elimination.
A longer term goal of this work is to develop an approach that is effective in practice for
quantified formulas involving both background theories and uninterpreted functions.
We plan to investigate the use of the framework described in this paper 
as a component of such an approach.

\paragraph{Acknowledgements}
We would like to thank Peter Baumgartner for his help with converting the benchmarks used in the evaluation
to the TPTP format.

{
\bibliographystyle{abbrv}
\bibliography{main}
}

\end{document}